\documentclass[11pt, letterpaper,reqno]{amsart}

\usepackage{amssymb}
\usepackage{graphicx}
\usepackage[hyphens]{url}
\usepackage{hyperref}
\usepackage{appendix}
\usepackage{dsfont}

\numberwithin{equation}{section}

\newtheorem{theorem}{Theorem}[section]
\newtheorem{assumption}{Assumption}[section]
\newtheorem{corollary}{Corollary}[section]
\newtheorem{lemma}{Lemma}[section]

\theoremstyle{definition}
\newtheorem{definition}{Definition}[section]
\theoremstyle{remark}
\newtheorem{remark}{Remark}[section]

\newcommand{\leref}{Lemma~\ref}
\newcommand{\deref}{Definition~\ref}

\newcommand{\thref}{Theorem~\ref}
\newcommand{\reref}{Remark~\ref}
\newcommand{\asref}{Assumption~\ref}

\renewcommand{\P}{\mathbb{P}}

\newcommand{\R}{\mathbb{R}}
\newcommand{\E}{\mathbb{E}}

\newcommand{\cF}{\mathcal{F}}
\newcommand{\F}{\mathbb{F}}

\newcommand{\cT}{\mathcal{T}}
\newcommand{\T}{\mathbb{T}}

\newcommand{\cP}{\mathcal{P}}
\newcommand{\cQ}{\mathcal{Q}}

\newcommand{\cL}{\mathcal{L}}

\newcommand{\eps}{\varepsilon}

\newcommand{\f}{\mathfrak{f}}

\newcommand{\EQUIV}{\Longleftrightarrow}
\newcommand{\cH}{\mathcal{H}}
\newcommand{\Phin}{\overline\Phi_{\alpha,\beta,\gamma}^n(\overline H,a,b,\overline\mu,c)}
\newcommand{\PhiN}{\overline\Phi_{\alpha,\beta,\gamma}^N(\overline H,a,b,\overline\mu,c)}
\newcommand{\PhiNN}{\overline\Phi_{\alpha,\beta,\gamma}^{N+1}(\overline H,a,b,\overline\mu,c)}
\newcommand{\tPhin}{\tilde\Phi_{\alpha,\beta,\gamma}^n(\tilde H,a,b,\tilde\mu,c)({\bf t})}
\newcommand{\tPhiN}{\tilde\Phi_{\alpha,\beta,\gamma}^N(\tilde H,a,b,\tilde\mu,c)({\bf t})}
\newcommand{\tPhiNN}{\tilde\Phi_{\alpha,\beta,\gamma}^{N+1}(\tilde H,a,b,\tilde\mu,c)({\bf t})}
\newcommand{\hPhin}{\hat\Phi_{\alpha,\beta,\gamma}^n(\hat H,a,b,\hat\mu,c)({\bf v})}
\newcommand{\hPhiN}{\hat\Phi_{\alpha,\beta,\gamma}^N(\hat H,a,b,\hat\mu,c)({\bf v})}
\newcommand{\hPhiNN}{\hat\Phi_{\alpha,\beta,\gamma}^{N+1}(\hat H,a,b,\hat\mu,c)({\bf v})}

\newcommand{\cR}{\mathcal{R}}

\title[]{No-arbitrage and hedging with liquid American options}
\author[]{Erhan Bayraktar}\thanks{E. Bayraktar is supported in part by the National Science Foundation under grant DMS-1613170 and by the
Susan M. Smith Professorship.}
\address{Department of Mathematics, University of Michigan}
\email{erhan@umich.edu}
\author[]{Zhou Zhou}
\address{Department of Mathematics, University of Michigan}
\email{zhouzhou@umich.edu}
\date{\today}
\keywords{semi-static trading strategies, Liquid American options, Fundamental theorem of asset pricing, sub/super hedging dualities}
\begin{document}
\maketitle

\begin{abstract}
Since most of the traded options on individual stocks is of American type it is of interest to generalize the results obtained in semi-static trading to the case when one is allowed to statically trade American options. However, this problem has proved to be elusive so far because of the asymmetric nature of the positions of holding versus shorting such options.
Here we provide a unified framework and generalize the fundamental theorem of asset pricing (FTAP) and hedging dualities in \cite{ZZ8} to the case where the investor can also short American options. Following \cite{ZZ8}, we assume that the longed American options are divisible. As for the shorted American options, we show that the divisibility plays no role regarding arbitrage property and hedging prices. Then using the method of enlarging probability spaces proposed in \cite{Tan}, we convert the shorted American options to European options, and establish the FTAP and sub- and super-hedging dualities in the enlarged space both with and without model uncertainty.
\end{abstract}

\section{Introduction}

Recently there has been some fundamental work on no-arbitrage and hedging in a financial market where stocks are traded dynamically and liquid options are traded statically (semi-static strategies), see e.g., \cite{Hobson1, Mathias, nutz2, Sch3} and the references therein. Even though, \cite{ZZ4, Hobson3, 2016arXiv160402274H, 2016arXiv160404608B, Tan}, consider the problem of hedging American options in this framework, it is worth noting that in all the above papers the liquid options are restricted to be European-style. But since most options on individual stocks are of American type, it is of practical interest to consider these problems when one can use the American options for hedging purposes.
So far only three papers considered American options as hedging devices: \cite{Campi} studies the completeness of the market where American put options of all the strike prices are available for semi-static trading, \cite{Cox} studies the no arbitrage conditions on the price function of American put options where European and American put options are available, \cite{ZZ8} considers FTAP and hedging duality with liquid American options.

The difficulty of using American options in semi-static trading lies in the asymmetric nature of positions of holding versus shorting this option. Our starting point in this paper is \cite{ZZ8}, where we assume that the liquid American options can only be bought, but not sold, and only the sub-hedging price but not the super-hedging price of the hedged American option is considered. The reason is that, if liquid American options are sold, or if the super-hedging is considered, then the investor needs to use a trading strategy that is adapted to the stopping strategy used by the holders of the American options. From this point of view, the problem becomes more complicated.
In this paper, we resolve these difficulties and generalize the FTAP and hedging results to the case where shorting American options is also allowed, hence creating a unified treatment of the problem. In particular, we assume that there are liquid American options that can be sold, and we also consider the super-hedging price of an American option. 

We assume that the longed American options are divisible, and as demonstrated by \cite[Section 2]{ZZ8}, this is a crucial assumption for obtaining the FTAP and subhedging duality.
We show in this paper that the divisibility plays no role for the shorted American options regarding arbitrage and hedging. Then using the method of enlarging probability spaces proposed in \cite{Tan} (also see \cite{Tannew} for its second version), we convert the shorted American options to European options. We establish the FTAP and sub- and super-hedging dualities for American options for models with and without model uncertainty.

A main contribution of this paper is that it provides a unified framework for the FTAP and hedging dualities when American options are available for both static buying and selling. There is extensive literature on FTAP and hedging duality with liquid European options. On the other hand, due to the flexibility of American options, conceptually and technically it is much more difficult to take liquid American options into consideration. Yet most of options in the market are of American style. As far as we know, before this paper there are only three papers \cite{Campi, Cox, ZZ8} consider American options being liquid options, yet none of them provides such a general framework, not to mention the FTAP and hedging duality results in this framework. 

We assume options are bounded. This is because our proof crucially relies on the compactness of liquidating strategies under the weak star topology (or Baxter-Chacon topology, see e.g., \cite{Edgar}). The boundedness of American options is needed in order to use this weak star topology to prove the FTAP and hedging results without model uncertainty, as well as apply some minimax argument for the proof of hedging dualities with model uncertainty. One novelty of this paper lies in the incorporation of liquid American options, and liquidating strategies for American options is crucial and also practical for the FTAP and hedging results. We think such boundedness assumption is not restrictive. First of all, what we have in mind for liquid options are put options, which are bounded (e.g., \cite{Cox} considers American put options). Second, we are in a finite discrete time set-up, and within finitely many time steps, even the stock prices could be reasonably assumed to be bounded, not to mention the option payoffs.

Another assumption is the continuity of options in the case of model uncertainty. Such assumption can be expected. First, a discretization argument is used for the proof of hedging dualities, and continuity is essential for the discretization. Second, for each (longed) American option we have infinitely many possible payoffs  associated with this American option, due to infinitely many possible liquidating strategies. In this sense, we can think of one American option as infinitely many European options (but with merely one price). For the existing literature on FTAP and hedging with infinitely many European options in a model-free or model uncertainty setup, some continuity assumption of options is often imposed. See e.g., \cite{Tan, Sch3, Mathias}. In fact, our continuity assumption is weaker in the sense that even though we assume American options at each period are continuous in $\omega$, their payoffs may still be discontinuous in $\omega$ because liquidating strategies may not be continuous. We again think such continuity assumption is not restrictive, since in practice most of options are (semi-)continuous.

The rest of the paper is organized as follows. In the next section, we first show that it makes no difference whether the shorted American options are divisible or not for the definitions of no-arbitrage and hedging prices. Then we work on an enlarged space and establish the FTAP and hedging dualities for a given model. In Section 3, we extend the FTAP and hedging dualities to the case of model uncertainty.

\section{No-arbitrage and hedging without model ambiguity}
In this section, we first describe the setup of our financial model without model ambiguity. We show that it makes no difference for arbitrage and hedging prices whether the shorted (not longed) American options are divisible or not. Then we reformulate the problems of arbitrage and hedging in an enlarged probability space, and establish FTAP and hedging dualities. Theorems \ref{p1}-\ref{tt1} are the main results of this section.
\subsection{Original probability space}
Let $(\Omega,\cF,\mathbb{F}=(\cF_t)_{t=0,1,\dotso,T},\P)$ be a filtered probability space, where $\cF$ is assumed to be separable, and $T\in\mathbb{N}$ represents the time horizon in discrete time. Let $S=(S_t)_{t=0,\dotso,T}$ be an adapted process taking values in $\R^d$ which represents the stock prices. Let $f^i:\Omega\mapsto\R$, $i=1,\dotso,L$, be $\cF_T$-measurable, representing the payoffs of European options. Let $g^j=(g_t^j)_{t=0,\dotso,T}$, $j=1,\dotso,M$; $h^k=(h^k_t)_{t=0,\dotso,T}$, $k=1,\dotso,N$, be $\mathbb{F}$-adapted processes, representing the payoff processes of American options. We assume that we can only buy but not sell each $f^i$ and $g^j$ at time $t=0$ with price $\alpha^i$ and $\beta^j$ respectively, and we can only sell but not buy each $h^k$ at time $t=0$ with price $\gamma^k$. Denote $f=(f^1,\dotso,f^L)$, $\alpha=(\alpha^1,\dotso, \alpha^L)$, and $\alpha-\eps=(\alpha^1-\eps,\dotso,\alpha^L-\eps)$ for scalar $\eps\in\R$. Similarly we will use $g,h$ and $\beta, \gamma$ for denoting vectors (of processes and prices). For simplicity, we assume that $f,g,h$ are bounded. Let $\phi=(\phi_t)_{t=0,\dotso,T}$ be $\mathbb{F}$-adapted, representing the payoff of the American options whose sub/superhedging prices we are interested in calculating by semi-statically trading in $S,f,g,h$. For simplicity, we assume that $\phi$ is bounded. We call $g$ and the sub-hedged $\phi$ longed American options, and $h$ and the super-hedged $\phi$ shorted American options. 

\begin{remark}
Here $f,g,h$ may represent the options whose trade is quoted with bid-ask spreads. For example, for an American option $\mathfrak{g}$ with bid price $\underline{\mathfrak{g}}$ and ask price $\overline{\mathfrak{g}}$, we can treat it as two American options, one that can only be bought at price $\overline{\mathfrak{g}}$, and the other that can only be sold at price $\underline{\mathfrak{g}}$.

We assume that the European options are only available to buy. This is in fact without loss of generality, because to short a European option $\mathfrak{f}$ is equivalent to long a European option $-\mathfrak{f}$. On the other hand, unlike European options, the treatments of American options that are bought and sold are very different, and that is the reason we separate American options into $g$ and $h$.
\end{remark}

We will often consider two cases, $n=N,N+1$, where $n$ represents the number of shorted American options. To be more specific, for $n=N$, the FTAP or sub-hedging will be considered, and there are $N$ shorted American options involved including $h^1,\dotso,h^N$. For $n=N+1$, the super-hedging of an American option will be considered, and there are $N+1$ shorted American options including $h^1,\dotso,h^N$ and the super-hedged American option.

If no American options are shorted (i.e.,  $h \equiv 0$ and we consider the sub-hedging $\phi$), then the only information the investor can observe is $\mathbb{F}$, and hence she will use a dynamic trading strategy in the stocks that is $\mathbb{F}$-adapted. Moreover, motivated by \cite[Section 2]{ZZ8}, we assume that the American options $g$ and $\phi$ (for sub-hedging) are divisible. That is, the investor can break each unit American options into pieces, and exercise each piece separately. We use the phrase \emph{liquidating strategy} to describe this type of exercise policy. it is more precisely defined as follows:

\begin{definition}
An $\mathbb{F}$-adapted process $\eta=(\eta_t)_{t=0,\dotso,T}$ is said to be an $\mathbb{F}$-liquidating strategy, if $\eta_t\geq 0$ for $t=0,\dotso,T$, and $\sum_{t=1}^T{\eta_t}=1$.
Denote $\cL$ as the set of all $\mathbb{F}$-liquidating strategies.
\end{definition}

Due to the existence of shorted American options $h$ and $\phi$ (recall that $\phi$ is shorted when we consider the problem of super-hedging), the investor's dynamic trading strategy for stock $S$ and liquidating strategy for longed American option $g^j$ should also be adapted to the stopping/liquidating strategies chosen by the holders of $h$ and $\phi$. Moreover, since options $g$  and sub-hedged $\phi$ are assumed to be divisible, it is natural to assume that options $h$ (and super-hedged $\phi$) are also divisible. This assumption makes practical sense:  the investor may sell $h^k$ to several agents, or sell to the same agent a large shares of $h^k$. However, as we will demonstrate, in terms of no-arbitrage and hedging prices (which we will define in the next section), it is in fact sufficient to let all shares of $h^k$ (and super-hedged $\phi$) be exercised once (i.e., holder of $h^k$ uses stopping times). 

\subsection{Discussion of divisibility for shorted American options for no-arbitrage and hedging}
In this sub-section we show that whether the shorted American options are divisible or not, the definitions of no-arbitrage and hedging prices coincide.

\subsubsection{Definitions of no arbitrage and hedging prices with divisibility}
Let $\cH$ be the set of $\mathbb{F}$-adapted processes taking values in $\R^d$. Let 
$$V:=\left\{(v_0,\dotso,v_T)\in\R_+^{T+1}:\ \sum_{t=0}^T v_t=1\right\},$$
which represents the space of liquidating strategies for each shorted American option. Let
$$\hat\cH^n:=\{\hat H(\cdot):V^n\mapsto\cH:\ \hat H_r(v^1,\dotso,v^n)=\hat H_r(u^1,\dotso,u^n),\text{ if }v_t^k=u_t^k \ \text{ for }\ t=0,\dotso,r; \ k=1,\dotso,n\},$$
and
$$\hat\cL^n:=\{\hat\eta(\cdot):V^n\mapsto\cL:\hat\eta _r(v^1,\dotso,v^n)=\hat\eta_r(u^1,\dotso,u^n),\text{ if }v_t^k=u_t^k \ \text{ for } t=0,\dotso,r;\ k=1,\dotso,n\},$$
where $v^k=(v_0^k,\dotso,v_T^k),u^k=(u_0^k,\dotso,u_T^k)\in V$ for $k=1,\dotso,n$, and $V^n$ is the $n$-fold Cartesian product of $V$.  Denote the set of semi-static trading strategies
$$\hat{\mathcal{A}}^n:=\{(\hat H,a,b,\hat\mu,c):\ (a,b,c)\in\R_+^L\times\R_+^M\times\R_+^N,\ \hat H\in\hat\cH^n,\ \hat\mu\in(\hat\cL^n)^M\}.$$
For $n=N,N+1$, denote the payoff using semi-static trading strategy $(\hat H,a,b,\hat\mu,c)\in\hat{\mathcal{A}}^n$  w.r.t. the prices $\alpha,\beta,\gamma$, with the realization ${\bf v}=(v^1,\dotso,v^n)\in V^n$ for shorted American options, 
\begin{equation}\label{e625}
\hPhin:=\hat H({\bf v})\cdot S+a(f-\alpha)+b\left(\left(\hat\mu({\bf v})\right)(g)-\beta\right)-\sum_{k=1}^N c^k(v^k(h^k)-\gamma^k),
\end{equation}
where for $H\in\cH$,
$$H\cdot S:=\sum_{t=0}^{T-1} H_t(S_{t+1}-S_t),$$
and for $\mu=(\mu^1,\dotso,\mu^M)\in(\cL^n)^M$,
$$\mu(g):=\left(\mu^1\left(g^1\right),\dotso,\mu^M\left(g^M\right)\right)$$
with
$$\mu^j\left(g^j\right):=\sum_{t=0}^T g_t^j\mu_t^j.$$
In the above equations we used denoted the inner product of vectors, say $x$ and $y$, by $xy$. Here at the right-hand-side of \eqref{e625}, the first term represents the payoff from trading stocks, second term the payoff from trading European options, third term the payoff from trading longed American options. and last term the payment for shorted American options.

\begin{definition}[No arbitrage]\label{d1}
We say NA$^1$ holds w.r.t. the prices $\alpha,\beta,\gamma$, if for any $(\hat H,a,b,\hat\mu,c)\in\hat{\mathcal{A}}^n$
$$\hPhiN\geq 0,\quad\P\text{-a.s. for any }{\bf v}\in V^N,$$
implies
$$\hPhiN= 0,\quad\P\text{-a.s. for any }{\bf v}\in V^N.$$
We say SNA$^1$ (SNA stands for ``strict no arbitrage'') holds, if there exists $\eps>0$ such that NA$^1$ holds w.r.t. the prices $\alpha-\eps,\beta-\eps,\gamma+\eps$.
\end{definition}

\begin{definition}[Hedging prices]
We define the sub-hedging price of $\phi$ by
\begin{eqnarray}
\notag\underline\pi^1(\phi)&:=&\sup\Big\{x\in\R:\ \exists(\hat H,a,b,\hat\mu,c)\in\hat{\mathcal{A}}^N\text{ and }\hat\eta\in\hat{\cL}^N,\\
\notag&&\text{s.t.}\ \hPhiN+(\hat\eta({\bf v}))(\phi)\geq x\ \P\text{-a.s., }\forall\,{\bf v}\in V^N\Big\},
\end{eqnarray}
and its super-hedging price by
\begin{eqnarray}
\notag\overline\pi^1(\phi)&:=&\inf\Big\{x\in\R:\ \exists(\hat H,a,b,\hat\mu,c)\in\hat{\mathcal{A}}^{N+1},\\
\notag&&\text{s.t.}\ x+\hPhiNN\geq v^{N+1}(\phi)\ \P\text{-a.s., }\forall\,{\bf v}=(v^1,\dotso,v^{N+1})\in V^{N+1}\Big\},
\end{eqnarray}
For a European option $\psi:\Omega\mapsto\R$, define its sub-hedging price as
$$\underline\pi_e^1(\psi):=\sup\left\{x\in\R:\ \exists(\hat H,a,b,\hat\mu,c)\in\hat{\mathcal{A}}^N,\ \text{s.t.}\ \hPhiN+\psi\geq x\ \P\text{-a.s., }\forall\,{\bf v}\in V^N\right\}.$$
\end{definition}
Let us clarify that $\exists$ should be applied before $\forall$ in the above definition.

\subsubsection{Definitions of no arbitrage and hedging prices without divisibility} For $n=N,N+1$, let
$$\tilde\cH^n:=\{\tilde H(\cdot):\mathbb{T}^n\mapsto\cH:\ \tilde H_r(t^1,\dotso,t^n)=\tilde H_r(s^1,\dotso,s^n)\text{ for }r<r^*\},$$
and
$$\tilde\cL^n:=\{\tilde\eta(\cdot):\mathbb{T}^n\mapsto\cL:\tilde\eta _r(t^1,\dotso,t^n)=\tilde\eta_r(s^1,\dotso,s^n)\text{ for }r<r^*\},$$
where $\T:=\{0,\dotso,T\}$, and
$$r^*=\inf_{k\in I}\left(s^k\wedge t^k\right)\quad\text{with}\quad I=\{i\in\{1,\dotso,n\}:\ s^i\neq t^i\}.$$
Denote the set of semi-static trading strategies by
$$\tilde{\mathcal{A}}^n:=\{(\tilde H,a,b,\tilde\mu,c):\ (a,b,c)\in\R_+^L\times\R_+^M\times\R_+^N,\ \tilde H\in\tilde\cH^n,\ \tilde\mu\in(\tilde\cL^n)^M\}.$$
and the payoff using semi-static trading strategy $(\tilde H,a,b,\tilde\mu,c)\in\tilde{\mathcal{A}}^n$  w.r.t. the prices $\alpha,\beta,\gamma$, with the realization ${\bf t}=(t^1,\dotso,t^n)\in\mathbb{T}^n$ for shorted American options by
$$\tPhin:=\tilde H({\bf t})\cdot S+a(f-\alpha)+b\left(\left(\tilde\mu({\bf t})\right)(g)-\beta\right)-\sum_{k=1}^N c^k(h_{t^k}^k-\gamma^k).$$

\begin{definition}[No-arbitrage]\label{d1}
We say NA$^2$ holds w.r.t. the prices $\alpha,\beta,\gamma$, if for any $(\tilde H,a,b,\tilde\mu,c)\in\tilde{\mathcal{A}}^n$
$$\tPhiN\geq 0,\quad\P\text{-a.s. for any }{\bf t}\in\mathbb{T}^N,$$
implies
$$\tPhiN= 0,\quad\P\text{-a.s. for any }{\bf t}\in\mathbb{T}^N.$$
We say SNA$^2$ holds, if there exists $\eps>0$ such that NA$^2$ holds w.r.t. the prices $\alpha-\eps,\beta-\eps,\gamma+\eps$.
\end{definition}

\begin{definition}[Hedging prices]
We define the sub-hedging price of $\phi$ as
\begin{eqnarray}
\notag\underline\pi^2(\phi)&:=&\sup\Big\{x\in\R:\ \exists(\tilde H,a,b,\tilde\mu,c)\in\tilde{\mathcal{A}}^N\text{ and }\tilde\eta\in\tilde{\cL}^N,\\
\notag&&\text{s.t.}\ \tPhiN+(\tilde\eta({\bf t}))(\phi)\geq x\ \P\text{-a.s., }\forall\,{\bf t}\in\T^N\Big\},
\end{eqnarray}
and its super-hedging price as
\begin{eqnarray}
\notag\overline\pi^2(\phi)&:=&\inf\Big\{x\in\R:\ \exists(\tilde H,a,b,\tilde\mu,c)\in\tilde{\mathcal{A}}^{N+1},\\
\notag&&\text{s.t.}\ x+\tPhiNN\geq\phi_{t^{N+1}}\ \P\text{-a.s., }\forall\,{\bf t}=(t^1,\dotso,t^{N+1})\in\T^{N+1}\Big\}.
\end{eqnarray}
For a European option $\psi:\Omega\mapsto\R$, define its sub-hedging price as
$$\underline\pi_e^2(\psi):=\sup\left\{x\in\R:\ \exists(\tilde H,a,b,\tilde\mu,c)\in\tilde{\mathcal{A}}^N,\ \text{s.t.}\ \tPhiN+\psi\geq x\ \P\text{-a.s., }\forall\,{\bf t}\in\T^N\right\}.$$
\end{definition}

\subsubsection{The Equivalence of the no-arbitrage definitions and the hedging prices}

\begin{theorem}\label{p1}
We have
$\overline\pi^1(\phi)=\overline\pi^2(\phi)$ and $\underline\pi^1(\phi)=\underline\pi^2(\phi)$ and $\underline\pi_e^1(\psi)=\underline\pi_e^2(\psi).$
\end{theorem}
\begin{proof}
For the simplicity of presentation, we will only show $\underline\pi_e^1(\psi)=\underline\pi_e^2(\psi)$ for $L=M=0$ and $N=2$. The proof can be very easily adapted for the more general case.

Since
 $\underline\pi_e^1(\psi)\leq\underline\pi_e^2(\psi)$ is clear, we focus on the reverse inequality. Let $x<\underline\pi_e^2(\psi)$, then there exists $(\tilde H,c^1,c^2)\in\tilde\cH\times\R_+\times\R_+$ such that for any $(t^1,t^2)\in\T^2$,
$$\tilde H(t^1,t^2)\cdot S-c^1(h_{t^1}^1-\gamma^1)-c^2(h_{t^2}^2-\gamma^2)+\psi\geq x,\quad\P\text{-a.s.}.$$

Define $\hat H:\ V^2\mapsto\cH$,
\begin{equation}\notag
\hat H_r({u,v})=\sum_{s=0}^T\sum_{t=0}^T u_s v_t \tilde H_r(s,t),\quad{u}=(u_0,\dotso,u_T),{v}=(v_0,\dotso,v_T)\in V.
\end{equation}
For any ${u,v,u',v'}\in V$, if for $r=0,\dotso,t^*$, $u_r=u_r'$ and $v_r=v_r'$, then
\begin{eqnarray}
\notag&&\hspace{-0.7cm}\hat H_{t^*}({u,v})=\sum_{s=0}^T\sum_{t=0}^T u_s v_t\tilde H_{t^*}(s,t)\\
\notag&&\hspace{-0.7cm}=\sum_{s=0}^{t^*}\sum_{t=0}^{t^*} u_s v_t\tilde H_{t^*}(s,t)+\sum_{s=t^*+1}^T\sum_{t=0}^{t^*}u_s v_t\tilde H_{t^*}(s,t)+\sum_{s=0}^{t^*}\sum_{t=t^*+1}^T u_s v_t\tilde H_{t^*}(s,t)+\sum_{s=t^*+1}^T\sum_{t=t^*+1}^T u_s v_t\tilde H_{t^*}(s,t)\\
\notag&&\hspace{-0.7cm}=\sum_{s=0}^{t^*}\sum_{t=0}^{t^*} u_s v_t\tilde H_{t^*}(s,t)+\sum_{s=t^*+1}^T u_s\sum_{t=0}^{t^*} v_t\tilde H_{t^*}(T,t)+\sum_{t=t^*+1}^T v_t\sum_{s=0}^{t^*} u_s\tilde H_{t^*}(s,T)+\sum_{s=t^*+1}^T  u_s\sum_{t=t^*+1}^T v_t\tilde H_{t^*}(T,T)\\
\notag&&\hspace{-0.7cm}=\sum_{s=0}^{t^*}\sum_{t=0}^{t^*} u_s' v_t'\tilde H_{t^*}(s,t)+\sum_{s=t^*+1}^T u_s'\sum_{t=0}^{t^*} v_t'\tilde H_{t^*}(T,t)+\sum_{t=t^*+1}^T v_t'\sum_{s=0}^{t^*} u_s'\tilde H_{t^*}(s,T)+\sum_{s=t^*+1}^T  u_s'\sum_{t=t^*+1}^T v_t'\tilde H_{t^*}(T,T)\\
\notag&&\hspace{-0.7cm}=\sum_{s=0}^{t^*}\sum_{t=0}^{t^*} u_s' v_t'\tilde H_{t^*}(s,t)+\sum_{s=t^*+1}^T\sum_{t=0}^{t^*}u_s' v_t'\tilde H_{t^*}(s,t)+\sum_{s=0}^{t^*}\sum_{t=t^*+1}^T u_s' v_t'\tilde H_{t^*}(s,t)+\sum_{s=t^*+1}^T\sum_{t=t^*+1}^T u_s' v_t'\tilde H_{t^*}(s,t)\\
\notag&&\hspace{-0.7cm}=\sum_{s=0}^T\sum_{t=0}^T u_s' v_t'\tilde H_{t^*}(s,t)=\hat H_{t^*}({ u',v'}),
\end{eqnarray}
where for the third and fifth equalities we use the non-anticipativity of $\tilde H$ (see the definition of $\tilde\cH$). This implies $\hat H\in\hat\cH^2$. Now for any ${u,v}\in V$,
\begin{eqnarray}
\notag&&\hat H(u,v)\cdot S-c^1({u}(h^1)-\gamma^1)-c^2({v}(h^2)-\gamma^2)+\psi\\
\notag&&=\sum_{r=0}^{T-1}\sum_{s=0}^T\sum_{t=0}^T u_s v_t \tilde H_r(s,t)(S_{r+1}-S_r)-c^1\sum_{s=0}^T u_s(h_s^1-\gamma^1)-c^2\sum_{t=0}^T v_t(h_t^2-\gamma^2)+\psi\\
\notag&&=\sum_{s=0}^T\sum_{t=0}^T u_s v_t\left[\tilde H(s,t)\cdot S-c^1(h_s^1-\gamma^1)-c^2(h_t^2-\gamma^2)+\psi\right]\\
\notag&&\geq x,\quad\P\text{-a.s.}.
\end{eqnarray}
This implies that $x\leq\underline\pi_e^1(\psi)$. By the arbitrariness of $x$, we have $\underline\pi_e^2(\psi)\leq \underline\pi_e^1(\psi)$.
\end{proof}

\begin{theorem}\label{p2}
SNA$^1$ and SNA$^2$ are equivalent.
\end{theorem}

\begin{proof}
Let SNA$^2$ hold. Denote the sub-hedging price of each $f^i$ (resp. sub-hedging price of each $g^j$, super-hedging price of each $h^k$) using stock and other liquid options by $\underline\pi_e'(f^i)$ (resp. $\underline\pi'(g^j)$, $\overline\pi'(h^k)$). We will not differentiate the hedging prices for type 1 and type 2 since by \thref{p1} they are the same. By SNA$^2$, there exists $\eps>0$, such that NA$^2$ holds w.r.t. the prices $\alpha-\eps,\beta-\eps,\gamma+\eps$. This implies
\begin{equation}\label{e12}
\alpha^i-\eps\geq\underline\pi_e'(f^i),\quad\beta^j-\eps\geq\underline\pi'(g^j),\quad\gamma^k+\eps\leq\overline\pi'(h^k).
\end{equation}

Suppose SNA$^1$ fails, then for any $m\in\mathbb{N}$ there would exist $(\hat H^m,a^m,b^m,\hat\mu^m,c^m)\in\hat{\mathcal{A}}^N$ such that 
\begin{equation}\label{e10}
\hat\Phi_{\alpha-\frac{1}{m},\beta-\frac{1}{m},\gamma+\frac{1}{m}}^N(\hat H^m,a^m,b^m,\hat\mu^m,c^m)({\bf v})\geq 0,\quad\P\text{-a.s.}\ \forall\,{\bf v}\in V^N,
\end{equation}
and ${\bf v}^m\in V^N$ such that
$$\P\left\{\hat\Phi_{\alpha-\frac{1}{m},\beta-\frac{1}{m},\gamma+\frac{1}{m}}^N(\hat H^m,a^m,b^m,\hat\mu^m,c^m)({\bf v}^m)>0\right\}>0.$$
If $a^m=0,b^m=0,c^m=0$, then we would have that
\begin{equation}\label{e9}
\hat H({\bf v}^m)\cdot S\geq 0,\ \P\text{-a.s.},\quad\text{and}\quad\P\{\hat H({\bf v}^m)\cdot S> 0\}>0.
\end{equation}
SNA$^2$ implies that for any $H\in\cH$, if $H\cdot S\geq 0\ \P$-a.s. then $H\cdot S=0\ \P$-a.s.. This contradicts \eqref{e9}. Therefore, at least one of $a_i^m,b_j^m,c_k^m$ is not zero. Denote
$$d^m:=\max\{a_i^m,b_j^m,c_k^m,\ i=1,\dotso,L,\ j=1,\dotso,M,\ k=1,\dotso,N\}>0.$$
By \eqref{e10},
\begin{equation}\label{e11}
\frac{1}{d^m}\hat\Phi_{\alpha,\beta,\gamma}^N(\hat H^m,a^m,b^m,\hat\mu^m,c^m)({\bf v})+\frac{L+M+N}{m}\geq 0,\quad\P\text{-a.s.}\ \forall\,{\bf v}\in V^N.
\end{equation}
For each $m$ at least one of $a_i^m/d^m,b_j^m/d^m,c_k^m/d^m$ is equal to $1$. Without loss of generality, (up to a sub-sequence) assume $a_1^m/d^m=1$. Then by \eqref{e11},
$$\underline\pi_e'(f^1)\geq\alpha^1-\frac{L+M+N}{m}.$$
Since $m$ is arbitrary we have that $\underline\pi_e'(f^1) \geq \alpha^{1}$,
which contradicts \eqref{e12}. This shows that SNA$^2$ implies SNA$^1$.

We can show that SNA$^1$ implies SNA$^2$ using a similar argument.
\end{proof}

\begin{remark}
In terms of arbitrage and hedging prices, divisibility are essential for $g$ and the sub-hedged $\phi$ as indicated by \cite[Section 2]{ZZ8}, but not essential for $h$ and the super-hedged $\phi$ as indicated by Theorems \ref{p1} and \ref{p2}. Therefore in the rest the paper we will assume that the shorted American options $h$ and super-hedged $\phi$ are not divisible.
\end{remark}

\subsection{Enlarged probability space}

We will follow the method in \cite{Tan} to reformulate the problems of arbitrage and hedging in an enlarged space. The advantage for working on the enlarged space is that the shorted American options become European options.

We will again let $n=N$ or $n=N+1$. The case when $n=N$ is for FTAP and sub-hedging $\phi$, i.e., either $\phi$ is not involved or the investor longs $\phi$. The case for $n=N+1$ is for super hedging $\phi$. Let $\overline\Omega^n:=\Omega\times\mathbb{T}^n$. Here $\mathbb{T}^N$ (resp. $\mathbb{T}^{N+1}$) represents the space of the exercise times for shorted American options $h$ (resp. $h$ and $\phi$ for super-hedging). For $k=1,\dotso,n$, let $\theta^k: \overline\Omega^n\mapsto\mathbb{T}$,
\begin{equation}\label{e901}
\theta^k(\overline\omega^n)=t^k,\quad \overline\omega^n=(\omega,t^1,\dotso,t^n)\in\overline\Omega^n.
\end{equation}
We extend $S$ from $\Omega$ to $\overline\Omega^n$, i.e., $\overline S_t^n(\overline\omega^n):=S_t(\omega)$ for $\overline\omega^n=(\omega,t^1,\dotso,t^n)\in\overline\Omega^n$. We similarly extend $f^i,g^j$ and $\phi$ (for sub-hedging, i.e., when $n=N$) and we denote the extensions as $\overline{f^i}^n,\overline{g^j}^n,\overline\phi^N$, respectively. For  $k=1,\dotso,N$, we extend $h^k$ from $\Omega$ to $\overline\Omega^n$,
$$\overline{h^k}^n(\overline\omega^n):=h_{t^k}^k(\omega),\quad\overline\omega^n=(\omega,t^1,\dotso,t^n)\in\overline\Omega^n.$$
Similarly, (for super-hedging) we extend $\phi$ from $\Omega$ to $\overline\Omega^{N+1}$,
$$\overline\phi^{N+1}\left(\overline\omega^{N+1}\right):=\phi_{t^{N+1}}(\omega),\quad\overline\omega^{N+1}=(\omega,t^1,\dotso,t^{N+1})\in\overline\Omega^{N+1}.$$

\begin{remark}\label{r1}
The extensions $\overline\phi^N$ and $\overline\phi^{N+1}$ serve different roles: $\overline\phi^N$ is an American option, which will be considered in the sub-hedging problem, while $\overline\phi^{N+1}$ is a European option, which will be considered in the super-hedging problem.
\end{remark}

Next, let us define the enlarged filtration. For $t=0,\dotso,T$,
$$\overline{\mathcal{F}}_t^n:=\sigma(\mathcal{F}_t\times\mathbb{T}^n,\{\theta^k\leq s\},s=0,\dotso,t,\ k=1,\dotso,n),$$ 
where $\mathcal{F}_t\times\mathbb{T}^n:=\{A\times\mathbb{T}^n:\ A\in\mathcal{F}_t\}$. Denote $\overline{\mathbb{F}}^n:=\left(\overline{\mathcal{F}}_t^n\right)_{t=0,\dotso,T}$.

Finally, for $n=N,N+1$, let $P^n$ be any probability measure on $(\mathbb{T}^n,\mathcal{B}(\mathbb{T}^n))$\footnote{We use $\mathcal{B}$ to identify the Borel sigma-algebra.} with full support. That is, for any $(t^1,\dotso,t^n)\in\T^n$
$$P^n(\{t^1,\dotso,t^n\})>0.$$
Let $\overline\P^n:=\P\otimes P^n$. $\overline{\mathbb{M}}^n:=(\overline\Omega^n,\overline{\mathcal{F}}_T^n,\overline{\mathbb{F}}^n,\overline\P^n)$ will serve as the enlarged filtered probability space.

Let $\overline{\mathcal{H}}^N$ (resp. $\overline{\mathcal{H}}^{N+1}$) be the set of $\overline{\mathbb{F}}^N$-adapted (resp. $\overline{\mathbb{F}}^{N+1}$-adapted) processes, representing the set of dynamic trading strategies for the stock based on information $\mathbb{F}$ as well as the exercise times of $h$ (resp. $h$ and $\phi$ for super-hedging). Similarly, let $\overline\cL^n$ be the set of $\overline{\mathbb{F}}^n$-liquidating strategies. Below we give the definition of semi-static trading strategies in the enlarged filtered probability space.

\begin{definition}
For $n=N,N+1$, a quintuplet $(\overline H,a,b,\overline\mu,c)$ is said to be an $\overline{\mathbb{M}}^n$-semi-static trading strategy, if $(a,b,c)\in\R_+^L\times\R_+^M\times\R_+^N$, $\overline H\in\overline{\mathcal{H}}^n$, and $\overline\mu\in(\overline\cL^n)^M$. Denote $\overline{\mathcal{A}}^n$ as the set of $\overline{\mathbb{M}}^n$-semi-static trading strategies.
\end{definition}

The payoff using semi-static trading strategy $(\overline H,a,b,\overline\mu,c)\in\overline{\mathcal{A}}^n$ w.r.t. the prices $\alpha,\beta,\gamma$ is given by
$$\Phin:=\overline H\cdot \overline S^n+a(\overline f^n-\alpha)+b(\overline\mu(\overline g^n)-\beta)-c(\overline h^n-\gamma).$$

\subsection{FTAP and hedging dualities}
\begin{definition}[No arbitrage]\label{d1}
For $n=N,N+1$, we say no arbitrage (NA) holds in $\overline{\mathbb{M}}^n$ w.r.t. the prices $\alpha,\beta,\gamma$, if for any $\overline{\mathbb{M}}^n$-semi-static trading strategy $(\overline H,a,b,\overline\mu,c)\in\overline{\mathcal{A}}^n$
$$\Phin\geq 0,\ \overline\P^n\text{-a.s.},\quad\quad\text{implies}\quad\Phin= 0,\ \overline\P^n\text{-a.s.}.$$
We say strict no arbitrage (SNA) holds in $\overline{\mathbb{M}}^n$, if there exists $\eps>0$ such that NA holds in $\overline{\mathbb{M}}^n$ w.r.t. the prices $\alpha-\eps,\beta-\eps,\gamma+\eps$.
\end{definition}

Obviously, we have the following.
\begin{corollary}
NA$^2$ and SNA$^2$ are equivalent to NA and SNA in $\overline{\mathbb{M}}^N$, respectively.
\end{corollary}

For $n=N,N+1$, we will denote the collection martingale measures on $(\overline\Omega^n,\overline{\mathcal{F}}_T^n,\overline{\mathbb{F}}^n)$ by $\mathcal{M}^n$. 
Define the subset of $\mathcal{M}^n$ that are equivalent to $\P^n$ and are satisfying the distributional constraints (that come from having to price the given option prices correctly) by
\begin{equation}\label{e1}
\overline\cQ^n:=\left\{Q\sim\overline\P^n:\ \overline S^n\ \text{is a $Q$-martingale},\ E_Q\left[\overline f^n\right]< \alpha,\ E_Q\left[\overline h^n\right]>\gamma,\sup_{\tau\in\overline\cT^n}E_Q\left[\overline g_\tau^n\right]<\beta \right\},
\end{equation}
where $\overline\cT^n$ represents the set of $\overline{\mathbb{F}}^n$-stopping times,
$$\sup_{\tau\in\overline\cT^n}E_Q[\overline g_\tau^n]:=\left(\sup_{\tau\in\overline\cT^n}E_Q\left[\overline{g_\tau^1}^n\right],\dotso,\sup_{\tau\in\overline\cT^n}E_Q\left[\overline{g_\tau^M}^n\right]\right),$$
and the inequalities above are understood component-wise. 
\begin{theorem}[FTAP]\label{t1}
For $n=N,N+1$, SNA in $\overline{\mathbb{M}}^n\ \EQUIV \overline\cQ^n\neq\emptyset$.
\end{theorem}
\begin{proof}
The result is implied by \cite[Theorem 3.1]{ZZ8}. Indeed, the probability space introduced in \cite[Section 3]{ZZ8} is general enough to apply it to the enlarged space $\overline{\mathbb{M}}^n$.
\end{proof}

\begin{remark}
Let us point out that the proof of \cite[Theorem 3.1]{ZZ8} uses a separating hyperplane argument. The boundedness of options is used for the proof of the closedness of a certain set in $\mathbb{L}^\infty$ under the weak star topology. In particular, the weak compactness of liquidating strategies is crucial for the proof of the closedness.
\end{remark}

\begin{remark}\label{r2}
Intuitively, SNA in $\overline{\mathbb{M}}^N$ should be equivalent to SNA in $\overline{\mathbb{M}}^{N+1}$, since the additional information coming from $\theta^{N+1}$ plays no role in terms of no-arbitrage. The equivalence can also be verified from \thref{t1}. 

Indeed, if SNA holds in $\overline{\mathbb{M}}^N$, then there exists $Q\in\overline\cQ^N$ by \thref{t1}. For $t=0,\dotso,T$, let $Q^t:=Q\otimes\delta_{\{t\}}$. Then it is easy to see that $Q^t \in \mathcal{M}^{N+1}$ satisfying the inequalities in \eqref{e1} with $n=N+1$. Let
$$Q':=\frac{1}{T+1}\sum_{t=0}^T Q^t=\frac{1}{T+1}Q\otimes\left(\delta_{\{0\}}+\dotso+\delta_{\{T\}}\right).$$
Then $$Q'\in\overline\cQ^{N+1}.$$ Therefore, SNA also holds in $\overline{\mathbb{M}}^{N+1}$ by \thref{t1}.

Conversely, if SNA holds in $\overline{\mathbb{M}}^{N+1}$, then there exists $R'\in\overline\cQ^{N+1}$ by \thref{t1}. Let $R$ be the restriction of $R'$ on $\left(\overline\Omega^N,\overline{\mathcal{F}}_T^N\right)$. Noting that $\sup_{\tau\in\overline\cT^N}E_R\left[\overline{g_\tau^M}^N\right]\leq \sup_{\tau\in\overline\cT^{N+1}}E_{R'}\left[\overline{g_\tau^M}^{N+1}\right]$, we have $R\in\overline\cQ^N$. This implies that SNA also holds in $\overline{\mathbb{M}}^N$.
\end{remark}

\begin{remark}
If on the original probability space $(\Omega,\mathcal{F},\mathbb{F},\P)$ there exists a martingale measure $Q\sim\P$ such that $E_Q[f]<\alpha$, $\sup_{\tau\in\cT}E_Q[g_\tau]<\beta$, and $\sup_{\tau\in\cT}E_Q[h_\tau]>\gamma$ (where $\cT$ is the set of $\mathbb{F}$-stopping times), then intuitively we expect that SNA would hold in $\overline{\mathbb{M}}^N$ (and thus $\overline{\mathbb{M}}^{N+1}$ by \reref{r2}). This can also be seen from \thref{t1}.

Indeed, for $k=1,\dotso,N$, let $\tau^k\in\cT$ be such that $E_Q[h_{\tau^k}^k]>\gamma^k$. Now define $Q'$ on $\left(\overline\Omega^N,\overline{\mathcal{F}}_T^N\right)$ by
$$Q'(A):=\int_{\overline\Omega^{N}}1_A(\omega,t^1,\dotso,t^N)\,Q(d\omega)\,d\delta_{\{(\tau^1(\omega),\dotso,\tau^N(\omega))\}},\quad A\in\overline{\mathcal{F}}_T^N.$$
Then it can be show that $Q' \in \mathcal{M}^N$ satisfying the inequalities in \eqref{e1} with $n=N$. Let $P$ be a probability measure on $(\mathbb{T}^N,\mathcal{B}(\mathbb{T}^N))$ with full support and let $Q'':=Q\otimes P$. Then it can be shown that $Q''\sim\overline\P^n$ and $Q'' \in \mathcal{M}^N$ satisfying $E_{Q''}[\overline f^n]< \alpha$ and $\sup_{\tau\in\overline\cT^n}E_{Q''}[\overline g_\tau^n]<\beta$. For $\lambda\in(0,1)$, let $Q_\lambda:=(1-\lambda)Q'+\lambda Q''$. Then we can show that $Q_\lambda\in\overline\cQ^N$ for $\lambda$ close to $0$ enough. Therefore, SNA holds in $\overline{\mathbb{M}}^N$.
\end{remark}

For the following definition it would be helpful to recall the difference between the extensions $\overline\phi^N$ and $\overline\phi^{N+1}$ (see \reref{r1}).

\begin{definition}
We define the sub-hedging price of $\phi$ by
\begin{equation}\label{e3}
\underline\pi(\phi):=\sup\left\{x\in\R:\ \exists(\overline H,a,b,\overline\mu,c)\in\overline{\mathcal{A}}^N\text{ and }\overline\eta\in\overline{\cL}^N,\ \text{s.t.}\ \PhiN+\overline\eta(\overline\phi^N)\geq x,\ \overline\P^N\text{-a.s.}\right\},
\end{equation}
and its super-hedging price of by
\begin{equation}\label{e4}
\overline\pi(\phi):=\inf\left\{x\in\R:\ \exists(\overline H,a,b,\overline\mu,c)\in\overline{\mathcal{A}}^{N+1},\ \text{s.t.}\ x+\PhiNN\geq\overline\phi^{N+1},\ \overline\P^{N+1}\text{-a.s.}\right\}.
\end{equation}
\end{definition}
Obviously, we have the following.
\begin{corollary}
$\overline\pi^2(\phi)=\overline\pi(\phi)\quad\text{and}\quad\underline\pi^2(\phi)=\underline\pi(\phi).$
\end{corollary}

Recall the martingale measure set $\overline\cQ^n$ defined in \eqref{e1}. We have the following sub- and super-hedging dualities without model ambiguity.

\begin{theorem}[Hedging dualities]\label{tt1}
Let SNA hold in $\overline{\mathbb{M}}^N$ (and thus $\overline{\mathbb{M}}^{N+1}$ by \reref{r2}). Then
$$\underline\pi(\phi)=\inf_{Q\in\overline\cQ^N}\sup_{\tau\in\overline\cT^N}E_Q\left[\overline\phi_\tau^N\right]\quad\text{and}\quad\overline\pi(\phi)=\sup_{Q\in\overline\cQ^{N+1}}E_Q\left[\overline\phi^{N+1}\right].$$
Moreover, there exist optimal sub- and super-hedging strategies.
\end{theorem}

\begin{proof}
The result is implied by \cite[Theorem 3.2]{ZZ8}, since the probability space introduced in \cite[Section 3]{ZZ8} is general enough to apply it to the enlarged space $\overline{\mathbb{M}}^n$.
\end{proof}

\begin{remark}
Let us point out that the hedging dualities in \cite[Theorem 3.2]{ZZ8} follows from \cite[Theorem 3.1]{ZZ8}. As for the existence of optimal hedging strategies, the Bax-Chacon topology (see e.g., \cite{Edgar}) of liquidating strategies is used and that is where the boundedness (or integrability) of options gets involved.
\end{remark}

\begin{remark}
It is possible that
$$\inf_{Q\in\overline\cQ^N}\sup_{\tau\in\overline\cT^N}E_Q\left[\overline\phi_\tau^N\right]>\sup_{\tau\in\overline\cT^N}\inf_{Q\in\overline\cQ^N}E_Q\left[\overline\phi_\tau^N\right].$$
We refer to \cite[Example 2.1]{ZZ4} for such an example. In fact, the right-hand-side of the above would be the sub-hedging price of $\phi$ if we assume $\phi$ is not divisible (see \cite[Theorem 2.1]{ZZ4}). Note that the divisibility matters for longed American options including the sub-hedged $\phi$.
\end{remark}

\begin{remark}
It can be shown that
\begin{equation}\label{e2}
\inf_{Q\in\overline\cQ^N}\sup_{\tau\in\overline\cT^N}E_Q\left[\overline\phi_\tau^N\right]\leq\sup_{Q\in\overline\cQ^N}\sup_{\tau\in\overline\cT^N}E_Q\left[\overline\phi_\tau^N\right]\leq\sup_{Q\in\overline\cQ^{N+1}}E_Q\left[\overline\phi^{N+1}\right].
\end{equation}
Indeed, for any $Q\in\overline\cQ^N$, let $\tau^\eps\in\overline\cT^N$ be an $\eps$ optimizer for $\sup_{\tau\in\overline\cT^N}E_Q\left[\overline\phi_\tau^N\right]$. Define a probability measure on $Q'$ on $\left(\overline\Omega^{N+1},\overline{\mathcal{F}}_T^{N+1}\right)$ by
$$Q'(A):=\int_{\overline\Omega^{N+1}}1_A(\omega,t^1,\dotso,t^N,t^{N+1})\,dQ(\omega,t^1,\dotso,t^N)\,d\delta_{\{\tau^\eps(\omega,t^1,\dotso,t^N)\}},\quad A\in\overline{\mathcal{F}}_T^{N+1}.$$
Then it can be shown that $Q' \in \mathcal{M}^{N+1}$ satisfying the inequalities in \eqref{e1} with $n=N+1$. Moreover,
$$E_{Q'}\left[\overline\phi^{N+1}\right]=E_Q\left[\overline\phi_{\tau^\eps}^N\right]\geq\sup_{\tau\in\overline\cT^N}E_Q\left[\overline\phi_\tau^N\right]-\eps.$$
Let
$$Q'':=\frac{1}{T+1}Q\otimes\left(\delta_{\{0\}}+\dotso+\delta_{\{T\}}\right).$$
Then it can be shown that $Q''\in\overline\cQ^{N+1}$. Let
$$Q_\lambda:=(1-\lambda)Q'+\lambda Q''.$$
Then it is easy to see that $Q_\lambda\in\overline\cQ^{N+1}$ for any $\lambda\in(0,1)$. Moreover, by choosing $\lambda$ close to $0$ enough, we can let $Q_\lambda$ be such that
$$\sup_{\hat Q\in\overline\cQ^{N+1}}E_{\hat Q}\left[\overline\phi^{N+1}\right]\geq E_{Q_\lambda}\left[\overline\phi^{N+1}\right]\geq E_{Q'}\left[\overline\phi^{N+1}\right]-\eps\geq\sup_{\tau\in\overline\cT^N}E_Q\left[\overline\phi_\tau^N\right]-2\eps.$$
By the arbitrariness of $Q$ and $\eps$, we have \eqref{e2} holds.

In fact, \eqref{e2} should hold on an intuitive level, because the second term in \eqref{e2} corresponds to the super-hedging price of $\phi$ in the case where the holder of $\phi$ reveals the stopping strategy at the beginning to the hedger. We refer to \cite[Section 3.1]{ZZ4} for a detailed discussion.

Let us also point out that in \eqref{e2} the third term may be strictly greater than the second term. We refer to \cite[Example 3.1]{ZZ4}, \cite{2016arXiv160402274H}, and \cite[Example 2.9]{Tan} for such examples.
\end{remark}

\section{No-arbitrage and hedging under model uncertainty}

In this section, we extend the FTAP and hedging dualities to the case of model uncertainty, which is described by a collection of measures which is not necessarily dominated. We will formulate the arbitrage and hedging directly on the enlarged space. 
Theorems \ref{t3} and \ref{t4} are the main results of this section. In this section, our notation will slightly change to accommodate the model uncertainty.

\subsection{Original space}
We follow the set-up in \cite{nutz2}. For any set $\mathbb{X}$, let $\mathcal{B}(\mathbb{X})$ be its Borel sigma algebra, and $\mathfrak{P}(\mathbb{X})$ be the set of probability measures on $(\mathbb{X},\mathcal{B}(\mathbb{X}))$. Let $\Xi$ be a complete separable metric space and $T\in\mathbb{N}$ be the time horizon. Let $\Xi_t:=\Xi^t$ be the $t$-fold Cartesian product for $t=1,\dotso,T$ (with convention $\Xi_0$ is a singleton). Denote $\Omega:=\Xi_T$. We denote by $\mathcal{F}_t$ the universal completion of $\mathcal{B}(\Xi_t)$, and $\mathbb{F}:=(\mathcal{F}_t)_{t=0,\dotso,T}$. For each $t\in\{0,\dotso,T-1\}$ and $\omega\in\Xi_t$, we are given a nonempty convex set of probability measures $\mathcal{P}_t(\omega)$ on $(\Xi,\mathcal{B}(\Xi))$. We assume that for each $t$, the graph of $\mathcal{P}_t$ is analytic, which ensures that $\mathcal{P}_t$ admits a universally measurable selector, i.e., a universally measurable kernel $P_t:\ \Xi_t\rightarrow \mathfrak{P}(\Xi)$ such that $P_t(\omega)\in\mathcal{P}_t(\omega)$ for all $\omega\in\Xi_t$. Let
\begin{equation}\label{prob}
\mathcal{P}:=\{P_0\otimes\dotso\otimes P_{T-1}:\ P_t(\cdot)\in\mathcal{P}_t(\cdot),\ t=0,\dotso,T-1\},
\end{equation}
where each $P_t$ is a universally measurable selector of $\mathcal{P}_t$, and
$$P_0\otimes\dotso\otimes P_{T-1}(A)=\int_{\Omega_1}\dotso\int_{\Omega_1} 1_A(\omega_1,\dotso,\omega_T)P_{T-1}(\omega_1,\dotso,\omega_{T-1};d\omega_T)\dotso P_0(d\omega_1),\ \ \ A\in\mathcal{B}(\Omega).$$

The concepts $S,f,g,h,\alpha,\beta,\gamma,\phi$ are defined as in Section 2, except that here we require $S,g,h,\phi$ to be $(\mathcal{B}(\Xi_t))_t$-adapted, $f$ to be $\mathcal{B}(\Omega)$-measurable.

We make the following standing assumption.
\begin{assumption}\label{a1}{\ }
\begin{itemize}
\item[(i)] The set of martingale measures on $(\Omega,\mathcal{F}_t,\mathbb{F})$
$$\cQ(\alpha):=\{Q\lll\cP:\ S\text{ is a $Q$-martingale},\ E_Q[f]\leq\alpha\}$$
is weakly compact, where $Q\lll\cP$ means that there exists some $P\in\cP$ dominating $Q$.
\item[(ii)] $f$ is bounded from below, $h$ is bounded from above and upper-semicontinuous in $\omega\in\Omega$.
\item[(iii)] $g$ and $\phi$ are bounded and uniformly continuous in $\omega\in\Omega$.
\end{itemize}
\end{assumption}
\begin{remark}
We refer to \cite[Examples 5.1 \& 5.2]{ZZ8} for examples satisfying \asref{a1}.
\end{remark}
\begin{remark}
If $\phi$ is considered to be super-hedged, then $\phi$ is only required to be bounded from above and upper-semicontinuous in $\omega\in\Omega$.
\end{remark}

\subsection{Enlarged space}
For $n=N,N+1$, as in Section 2.3, let $(\overline\Omega^n,\overline{\mathcal{F}}_T^n,\overline{\mathbb{F}}^n)$ be the enlarged filtered space for $(\Omega,\mathcal{F}_T,\mathbb{F})$. All the variables $\overline f^n,\overline g^n,\overline h^n,\overline\phi^n,\overline{\mathcal{H}}^n,\overline\cL^n, \overline{\mathcal{A}}^n,\Phin$ are as before, except that we require that for any $\overline\mu=(\overline\mu_0,\dotso,\overline\mu_T)\in\overline\cL^n$, $\overline\mu_t$ is Borel measurable. 

Let
$$\overline\cP^n:=\{P\otimes R:\ P\in\cP,\ R\text{ is a probability measure on }(\mathbb{T}^n,\mathcal{B}(\mathbb{T}^n))\}, $$
and denote $\overline{\mathbb{M}}^n:=(\overline\Omega^n,\overline{\mathcal{F}}_T^n,\overline{\mathbb{F}}^n,\overline\cP^n)$.

\subsection{FTAP and hedging dualities}
\begin{definition}[No arbitrage]\label{d2}
For $n=N,N+1$, we say no arbitrage (NA) holds in $\overline{\mathbb{M}}^n$ w.r.t. the prices $\alpha,\beta,\gamma$, if for any $\overline{\mathbb{M}}^n$-semi-static trading strategy $(\overline H,a,b,\overline\mu,c)\in\overline{\mathcal{A}}^n$
$$\Phin\geq 0,\ \overline\cP^n\text{-q.s.}\footnote{For a set of probability measures $\cP$, we say a property holds $\cP$-q.s., if the property holds $P$-a.s. for any $P\in\cP$.},\quad\quad\text{implies}\quad\Phin= 0,\ \overline\cP^n\text{-q.s.}.$$
We say strict no arbitrage (SNA) holds in $\overline{\mathbb{M}}^n$, if there exists $\eps>0$ such that NA holds in $\overline{\mathbb{M}}^n$ w.r.t. the prices $\alpha-\eps,\beta-\eps,\gamma+\eps$.
\end{definition}

\begin{definition}[Hedging prices]
We define the sub-hedging price of $\phi$ by
\begin{equation}\label{e5}
\underline\pi(\phi):=\sup\left\{x\in\R:\ \exists(\overline H,a,b,\overline\mu,c)\in\overline{\mathcal{A}}^N\text{ and }\overline\eta\in\overline{\cL}^N,\ \text{s.t.}\ \PhiN+\overline\eta(\overline\phi^N)\geq x,\ \overline\cP^N\text{-q.s.}\right\},
\end{equation}
and its super-hedging price by
\begin{equation}\label{e6}
\overline\pi(\phi):=\inf\left\{x\in\R:\ \exists(\overline H,a,b,\overline\mu,c)\in\overline{\mathcal{A}}^{N+1},\ \text{s.t.}\ x+\PhiNN\geq\overline\phi^{N+1},\ \overline\cP^{N+1}\text{-q.s.}\right\}.
\end{equation}
\end{definition}
For $n=N,N+1$ and $(\alpha',\beta',\gamma')\in\R^L\times\R^M\times\R^N$, define the set of martingale measures on $(\overline\Omega^n,\overline{\mathcal{F}}_T^n,\overline{\mathbb{F}}^n)$,
\begin{equation}\label{e29}
\overline\cQ^n(\alpha',\beta',\gamma'):=\left\{Q\lll\overline\cP^n:\ \overline S^n\text{ is a $Q$-martingale},\ E_Q\left[\overline f^n\right]\leq\alpha,\ E_Q\left[\overline h^n\right]\geq\gamma,\ \sup_{\tau\in\overline\cT^n}E_Q\left[\overline g_\tau^n\right]\leq\beta\right\}.
\end{equation}

\begin{theorem}[Hedging dualities]\label{t3}
Let \asref{a1} hold. If SNA holds in $\overline{\mathbb{M}}^N$, then
$$\underline\pi(\phi)=\inf_{Q\in\overline\cQ^N(\alpha,\beta,\gamma)}\sup_{\tau\in\overline\cT^N}E_Q\left[\overline\phi_\tau^N\right]\quad\text{and}\quad\overline\pi(\phi)=\sup_{Q\in\overline\cQ^{N+1}(\alpha,\beta,\gamma)}E_Q\left[\overline\phi^{N+1}\right].$$
Moreover, there exist $Q'\in\overline\cQ^N(\alpha,\beta,\gamma)$ and $Q''\in\overline\cQ^{N+1}(\alpha,\beta,\gamma)$ that attain the infimum and supremum for the dualities above.
\end{theorem}

\begin{theorem}[FTAP]\label{t4}
Let \asref{a1} hold. Then SNA holds in $\overline{\mathbb{M}}^N$ if and only if there exists $\eps>0$ such that for any $P\in\overline\cP^N$, there exists $Q\in\overline\cQ^N(\alpha-\eps,\beta-\eps,\gamma+\eps)$ dominating $P$.
\end{theorem}

\subsection{Proofs of Theorems \ref{t3} and \ref{t4}}

Before we give the proofs of the main results of this section, we will obtain some preliminary results. The proofs of the theorems are at the end of this sub-section. Throughout we use $n$, when we mean that a statement hold for both $n=N$ and $n=N+1$.

\begin{definition}
We say NA$(\overline\cP^n)$ holds, if for any $\overline H\in\overline\cH^n$,
$$\overline H\cdot\overline S\geq 0,\ \overline\cP^n\text{-q.s.},\quad\quad\text{implies}\quad\overline H\cdot\overline S=0,\ \overline\cP^n\text{-q.s.}.$$
\end{definition}
Define the set of martingale measures on $(\overline\Omega^n,\overline{\mathcal{F}}_T^n,\overline{\mathbb{F}}^n)$,
$$\overline\cQ^n:=\{Q\lll\overline\cP^n:\ \overline S^n\text{ is a $Q$-martingale}\}.$$

\begin{lemma}\label{l4}
NA$(\overline\cP^n)$ holds if and only if for any $P\in\overline\cP^n$ there exists $Q\in\overline\cQ^n$ dominating $P$.
\end{lemma}
\begin{proof}
The results of \cite{nutz2} does not directly apply because they work on the canonical space. But extending their result is not difficult. First observe that as usual the sufficiency is obvious. Let us focus on the necessity. If NA$(\overline\cP^n)$ holds, then it is easy to see that NA$(\cP)$ holds in the original space $(\Omega,\mathcal{F}_T,\mathbb{F})$ (i.e., for any $\mathbb{F}$-adapted process $H$, if $H\cdot S\geq 0\ \cP$-q.s., then $H\cdot S=0\ \cP$-q.s.). Then for any $P\otimes R\in\overline\cP^n$ with $P\in\cP$ and $R$ a probability measure on $(\mathbb{T}^n,\mathcal{B}(\mathbb{T}^n))$, by \cite[Theorem 4.5]{nutz2} there exists $Q$ being a martingale measure on the original space $(\Omega,\mathcal{F}_T,\mathbb{F})$ and $P'\in\cP$ such that $P\ll Q\ll P'$. Then $Q\otimes R\in \overline\cQ^n$ dominates $P\otimes R$.
\end{proof}

Let $\zeta:\overline\Omega^n\mapsto\R$ be Borel measurable. Define the super-hedging price of $\zeta$ using only dynamic trading in $S$ as
$$\pi^0(\zeta):=\inf\{x\in\R:\ \exists\overline H\in\overline\cH^n,\ \text{s.t. }x+\overline H\cdot\overline S^n\geq\zeta,\ \overline\cP^n\text{-q.s.}\}.$$

\begin{lemma}\label{l5}
Let NA$(\overline\cP^n)$ hold. Then
\begin{equation}\label{e905}
\pi^0(\zeta)=\sup_{Q\in\overline\cQ^n}E_Q[\zeta].
\end{equation}
Moreover, there exists an optimal super-hedging strategy.
\end{lemma}
\begin{proof}
We adapt the arguments in \cite[Section 3.1]{Tan}, and will only sketch the proof. For simplicity of the presentation, we assume that $n=2$. For $t=0,\dotso,T$, let
$$\overline\Xi_t=\Xi_t\times\{0,\dotso,t\}^2\quad\text{and}\quad\overline{\mathcal{G}}_t:=\cF_t\vee\sigma(\theta^1\wedge t)\vee\sigma(\theta^2\wedge t),$$
where $\theta^1$ and $\theta^2$ are defined as in \eqref{e901}.  For $t=0,\dotso,T-1$ and $\omega\in\Xi_t$, let
$$\cQ_t(\omega):=\{Q\lll\cP_t(\omega):\ \E_Q[y(S_{t+1}(\omega,\cdot)-S_t(\omega))]=0,\ \forall y\in\R^d\}.$$
We provide the main idea of the proof in the following four steps.

\textbf{Step 1}. Let $t\in\{0,\dotso,T-1\}$. For a upper-semianalytic function $\chi:\overline\Xi^{t+1}\mapsto\bar\R$, define the map $\mathcal{E}_t(\chi):\overline\Xi^t\mapsto\bar\R$ by
\begin{eqnarray}
\notag&&\mathcal{E}_t(\chi)(\overline\omega):=\sup_{Q\in\cQ_t(\omega)}\bigg\{E_Q[\chi(\omega,\cdot,s^1,s^2)]1_{\{s^1<t,\,s^2<t\}}\\
\notag &&+E_Q[\chi(\omega,\cdot,s^1,t)]\vee E_Q[\chi(\omega,\cdot,s^1,t+1)]1_{\{s^1<t,\,s^2=t\}}\\
\notag &&+E_Q[\chi(\omega,\cdot,t,s^2)]\vee E_Q[\chi(\omega,\cdot,t+1,s^2)]1_{\{s^1=t,\,s^2<t\}}\\
\notag &&+E_Q[\chi(\omega,\cdot,t,t)]\vee E_Q[\chi(\omega,\cdot,t,t+1)]\vee E_Q[\chi(\omega,\cdot,t+1,t)]\vee E_Q[\chi(\omega,\cdot,t+1,t+1)]1_{\{s^1=s^2=t\}}\bigg\},
\end{eqnarray}
for $\overline\omega=(\omega,s^1,s^2)\in\overline\Xi^t$. From \cite[Lemma 4.10]{nutz2}, $\mathcal{E}_t(\chi)$ is upper-semianalytic. Moreover, there exist universally measurable functions $y^i:\overline\Xi^t\mapsto\R^d,\ i=1,2,3,4$ such that
\begin{eqnarray}
\notag\mathcal{E}_t(\chi)(\overline\omega)+y^1(\overline\omega)(S_{t+1}(\omega,\cdot)-S_t(\omega))&\geq&\chi(\omega,\cdot,s^1,s^2),\\
\notag\mathcal{E}_t(\chi)(\overline\omega)+y^2(\overline\omega)(S_{t+1}(\omega,\cdot)-S_t(\omega))&\geq&\chi(\omega,\cdot,s^1,t+1),\\
\notag\mathcal{E}_t(\chi)(\overline\omega)+y^3(\overline\omega)(S_{t+1}(\omega,\cdot)-S_t(\omega))&\geq&\chi(\omega,\cdot,t+1,s^2),\\
\notag\mathcal{E}_t(\chi)(\overline\omega)+y^4(\overline\omega)(S_{t+1}(\omega,\cdot)-S_t(\omega))&\geq&\chi(\omega,\cdot,t+1,t+1),
\end{eqnarray}
$\cP(\omega)$-q.s. for any $\overline\omega=(\omega,s^1,s^2)\in\overline\Xi^t$ such that NA$(\cP_t(\omega))$\footnote{That is, for any $y\in\R^d$, if $y(S_{t+1}(\omega,\cdot)-S_t(\omega))\geq 0\ \cP_t(\omega)$-q.s., then $y(S_{t+1}(\omega,\cdot)-S_t(\omega))=0\ \cP_t(\omega)$-q.s..} holds.

\textbf{Step 2}. Now assume $\zeta$ is bounded from above. Then following the argument in \cite[Lemma 3.2]{Tan} we can show that
\begin{equation}\label{e903}
\sup_{Q\in\overline\cQ_\text{loc}^2}E_Q[\zeta]=\mathcal{E}_0\circ\dotso\circ\mathcal{E}_{T-1}(\zeta),
\end{equation}
where
$$\cQ_\text{loc}^2:=\{Q\lll\overline\cP^2: \overline S^2\text{ is a $(Q,\overline\F^2)$-local martingale}\}.$$
Indeed, for any $Q\in\overline\cQ_\text{loc}^2$, we can use a backward induction to show that
\begin{equation}\label{e904}
E_Q[\zeta|\overline{\mathcal{G}}_t]\leq\mathcal{E}_t\circ\dotso\circ\mathcal{E}_{T-1}(\zeta)=:\mathcal{E}^t(\zeta)
\end{equation}
$Q$-a.s. for $t=0,\dotso,T-1$. This implies ``$\leq$'' holds for \eqref{e903}. Conversely, we can perform a measurable selection argument and construct an $\eps$-optimizer to show ``$\geq$'' for \eqref{e903}. 

\textbf{Step 3}. Assume $\zeta$ is bounded from above. Recall $\mathcal{E}^t(\zeta)$ defined in \eqref{e904} and denote $\mathcal{E}^T[\zeta]=\zeta$. By Step 1, for $t=0,\dotso,T-1$ there exist universally measurable functions $y_t^i,\ i=1,2,3,4$ such that
\begin{eqnarray}
\notag y_t^1(\overline\omega)(S_{t+1}(\omega,\cdot)-S_t(\omega))&\geq&\mathcal{E}^{t+1}(\zeta)(\omega,s^1,s^2)-\mathcal{E}^t(\zeta)(\overline\omega),\\
\notag y_t^2(\overline\omega)(S_{t+1}(\omega,\cdot)-S_t(\omega))&\geq&\mathcal{E}^{t+1}(\zeta)(\omega,s^1,t+1)-\mathcal{E}^t(\zeta)(\overline\omega),\\
\notag y_t^3(\overline\omega)(S_{t+1}(\omega,\cdot)-S_t(\omega))&\geq&\mathcal{E}^{t+1}(\zeta)(\omega,t+1,s^2)-\mathcal{E}^t(\zeta)(\overline\omega),\\
\notag y_t^1(\overline\omega)(S_{t+1}(\omega,\cdot)-S_t(\omega))&\geq&\mathcal{E}^{t+1}(\zeta)(\omega,t+1,t+1)-\mathcal{E}^t(\zeta)(\overline\omega),
\end{eqnarray}
$\cP(\omega)$-q.s. for any $\overline\omega=(\omega,s^1,s^2)\in\overline\Xi^t$ with NA$(\cP_t(\omega))$ holds. Therefore,
$$\sum_{t=0}^{T-1}\overline H_t(\overline S_{t+1}^2-\overline S_t^2)\geq\sum_{t=0}^{T-1}\left[\mathcal{E}^{t+1}(\zeta)-\mathcal{E}^t(\zeta)\right]=\zeta-\mathcal{E}^0[\zeta]=\zeta-\sup_{Q\in\overline\cQ_\text{loc}^2}E_Q[\zeta],\quad\overline\cP^2\text{-q.s.},$$
where for $\overline\omega=((\omega_1,\dotso,\omega_T),s^1,s^2)\in\overline\Omega^2$,
$$\overline H_t(\overline\omega):=y_t^1([\overline\omega]_t) 1_{\{s^1\leq t,\,s^2\leq t\}}+y_t^2([\overline\omega]_t) 1_{\{s^1\leq t,\,s^2> t\}}+y_t^3([\overline\omega]_t) 1_{\{s^1> t,\,s^2\leq t\}}+y_t^4([\overline\omega]_t) 1_{\{s^1> t,\,s^2>t\}},$$
with $[\overline\omega]_t:=((\omega_1,\dotso,\omega_t),s^1,s^2)\in\overline\Xi_t$ for $t=0,\dotso,T$. Hence,
$$\pi^0(\zeta)\leq\sup_{Q\in\overline\cQ_\text{loc}^2}E_Q[\zeta].$$
Using the proof in \cite[Lemma A.3]{nutz2}, we can show that $\sup_{Q\in\overline\cQ_\text{loc}^2}E_Q[\zeta]=\sup_{Q\in\overline\cQ^2}E_Q[\zeta]$. It is easy to show the weakly duality,
$$\pi^0(\zeta)\geq\sup_{Q\in\overline\cQ^2}E_Q[\zeta].$$
Therefore, we have \eqref{e905} holds when $\zeta$ is bounded from above.

\textbf{Step 4}. In general, using \cite[Theorem 2.2]{nutz2} we can show that
$$\lim_{m\rightarrow\infty}\pi^0(\zeta\wedge m)=\pi^0(\zeta).$$
Moreover, by monotone convergence theorem,
$$\lim_{m\rightarrow\infty}\sup_{Q\in\overline\cQ^2}E_Q[\zeta\wedge m]=\sup_{Q\in\overline\cQ^2}E_Q[\zeta].$$
Hence \eqref{e905} holds. The existence of an optimal hedging strategy follows from \cite[Theorem 2.3]{nutz2}.
\end{proof}

Let $\xi:=(\xi^1,\dotso,\xi^e):\overline\Omega^n\mapsto\R^e$ be Borel measurable, representing $e$ European options available for static buying with prices $\tilde\xi:=(\tilde\xi^1,\dotso,\tilde\xi^e)\in\R^e$ in the space $(\overline\Omega^n,\overline{\mathcal{F}}_T^n,\overline{\mathbb{F}}^n)$. For $i=1,\dotso,e$, assume $E_Q|\xi^i|<\infty$ and $E_Q|\zeta|<\infty$ for any $Q\in\overline\cQ^n$.

\begin{definition}
We say NA$(\overline\cP^n)$ holds with $(\xi,\tilde\xi)$ (in the space $(\overline\Omega^n,\overline{\mathcal{F}}_T^n,\overline{\mathbb{F}}^n)$), if for any $(\overline H,a)\in\overline\cH^n\times\R_+^e$,
$$\overline H\cdot\overline S^n+a(\xi-\tilde\xi)\geq 0,\ \overline\cP^n\text{-q.s.},\quad\quad\text{then}\quad\overline H\cdot\overline S^n+a(\xi-\tilde\xi)=0,\ \overline\cP^n\text{-q.s.}.$$
We say SNA$(\overline\cP^n)$ holds with $(\xi,\tilde\xi)$, if there exists $\eps>0$ such that NA$(\overline\cP^n)$ holds with $(\xi,\tilde\xi-\eps)$.
\end{definition}
Given $\hat\xi\in\R^e$, let
$$\overline\cQ_{\hat\xi}^n:=\{Q\lll\overline\cP^n:\ \overline S^n\text{ is a $Q$-martingale},\ E_Q[\xi]\leq\hat\xi\}.$$
Define the super-hedging price of $\zeta$ with $(\xi,\hat\xi)$ as
$$\pi^e(\zeta):=\inf\{x\in\R:\ \exists(\overline H,a)\in\overline\cH^n\times\R_+^e,\ \text{s.t. }x+\overline H\cdot\overline S^n+a(\xi-\hat\xi)\geq\zeta,\ \overline\cP^n\text{-q.s.}\}.$$

\begin{lemma}\label{l2}
\begin{itemize}
\item[(i)] SNA$(\overline\cP^n)$ holds with $(\xi,\tilde\xi)$ if and only if there exists $\eps>0$, such that for any $P\in\overline\cP^n$, there exists $Q\in\overline\cQ_{\tilde\xi-\eps}^n$ dominating $P$.
\item[(ii)] Let SNA$(\overline\cP^n)$ hold with $(\xi,\tilde\xi)$. Then
\begin{equation}\label{e906}
\pi^e(\zeta)=\sup_{Q\in\overline\cQ_{\tilde\xi}^n}E_Q[\zeta].
\end{equation}
Moreover, there exists an optimal super-hedging strategy.
\end{itemize}
\end{lemma}
\begin{proof}
We adapt the arguments for \cite[Theorem 2.1]{ZZ3}, and will only sketch the proof.

Following Part 1 of the proof of \cite[Theorem 2.1]{ZZ3}, we can show that the set
$$\mathcal{C}:=\{\overline H\cdot\overline S^n+a(\xi-\tilde\xi)-W: \overline H\in\overline\cH^n,\,a\in\R_+^e,\,W\geq 0\ \overline\cP^n\text{-q.s.}\},$$
is closed. That is, if $U^j\in\mathcal{C}$ and $U^j\rightarrow U\ \overline\cP^n$-q.s. as $j\rightarrow\infty$, then $U\in\mathcal{C}$. Let $(H^j,a^j)\in\overline\cH^n\times\R_+^e$ such that
$$\pi^e(\zeta)+\frac{1}{j}+H^j\cdot\overline S^n+a^j(\xi-\tilde\xi)\geq\zeta,\quad\overline\cP^n\text{-q.s.}.$$
Then by the closedness of $\mathcal{C}$, there exists $(H,a)\in\overline\cH^n\times\R_+^e$ such that
$$\pi^e(\zeta)+H\cdot\overline S^n+a(\xi-\tilde\xi))\geq\zeta,\quad\overline\cP^n\text{-q.s.}.$$
This implies the existence of an optimal hedging strategy. 

We will use an induction to show (i) and \eqref{e906}. By Lemmas \ref{l4} and \ref{l5} the results hold for $e=0$. Assume the results hold for $e=k$, and consider the case when $e=k+1$. For $j=k,k+1$, denote $\pmb\xi^j=(\xi^1,\dotso,\xi^j)$, $\tilde{\pmb\xi^j}=(\tilde\xi^1,\dotso,\tilde\xi^j)$, $\pi^j(\cdot)$ the super-hedging price using stocks $\overline S^n$ and options $\pmb\xi^j$, and
$$\overline\cQ_{\hat{\pmb\xi}^j}^n:=\{Q\lll\overline\cP^n:\ \overline S^n\text{ is a $Q$-martingale},\ E_Q[\pmb\xi^j]\leq\hat{\pmb\xi}^j\},\quad\hat{\pmb\xi}^j\in\R^j.$$

\textbf{Proof of (i) when $e=k+1$}. The necessity part is easy to show, and we focus on the sufficiency. Since SNA holds with $(\pmb\xi^{k+1},\tilde{\pmb\xi}^{k+1})$, there exists $\eps>0$ such that SNA still holds with $(\pmb\xi^{k+1},\tilde{\pmb\xi}^{k+1}-\eps)$. Then by induction hypothesis, we have
$$\inf_{Q\in\overline\cQ_{\tilde{\pmb\xi}^k-\eps}^n}E_Q[\xi^{k+1}]=-\pi^k(-\xi^{k+1})\leq\tilde\xi^{k+1}-\eps.$$
Hence, there exists $Q^*\in\overline\cQ_{\tilde{\pmb\xi}^k-\eps}^n$ such that
\begin{equation}\label{e907}
E_{Q^*}[\xi^{k+1}]\leq\tilde\xi^{k+1}-\frac{\eps}{2}.
\end{equation}
As SNA also holds with $(\pmb\xi^{k},\tilde{\pmb\xi}^{k})$, by induction hypothesis there exists $\eps'>0$, such that for any $P\in\overline\cP^n$ there exists $Q\in\overline\cQ_{\tilde{\pmb\xi}^k-\eps'}^n$ dominating $P$. Now let $\delta:=\frac{\eps}{4}\wedge\frac{\eps'}{2}$. For any $P'\in\overline\cP^n$, let $Q'\in\overline\cQ_{\tilde{\pmb\xi}^k-\eps'}^n\subset\overline\cQ_{\tilde{\pmb\xi}^k-\delta}^n$ dominate $P'$.  Then by choosing $\lambda\in(0,1)$ and close to $0$ enough, we can show that $P'\ll(1-\lambda)Q^*+\lambda Q'\in\overline\cQ_{\tilde{\pmb\xi}^k-\delta}^n$.

\textbf{Proof of \eqref{e906} when $e=k+1$}. Without loss of generality, we may assume that $\zeta$ is bounded from above. Otherwise we can first consider $\zeta\wedge m$ and then send $m\rightarrow\infty$ as in Step 4 in the proof of \leref{l5}. It is easy to show ``$\geq$'' for \eqref{e906}, and we focus on the reverse inequality. It suffices to show that there exists $Q^j\in\overline\cQ_{\tilde{\pmb\xi}^k}^n$ such that
$$\lim_{j\rightarrow\infty}E_{Q^j}[\xi^{k+1}]\leq\tilde\xi^{k+1}\quad\text{and}\quad\lim_{j\rightarrow\infty}E_{Q^j}[\zeta]\geq\pi^{k+1}(\zeta).$$
Indeed, if such $Q^j$ exists, then one can take $Q_\lambda:=(1-\lambda)Q^j+\lambda Q^*$, where $\lambda\in(0,1)$ and $Q^*$ is chosen in \eqref{e907}; By choosing $j$ large enough and $\lambda$ close to $0$, we can show that $Q_\lambda\in\overline\cQ_{\tilde{\pmb\xi}^{k+1}}^n$ and $E_{Q_\lambda}[\zeta]$ arbitrarily close to $\pi^{k+1}(\zeta)$, which implies ``$\leq$'' for \eqref{e906}. 

Suppose such $Q^j$ does not exist, then
$$\overline{\left\{\left(E_Q[\xi^{k+1}],E_Q[\zeta]\right):\ Q\in\overline\cQ_{\tilde{\pmb\xi}^k}^n\right\}}\cap(-\infty,\tilde\xi^{k+1}]\times[\pi^{k+1},\infty)=\emptyset.$$
Then there exists $(y,z)\in\R^2\setminus\{(0,0)\}$ such that
$$\inf_{Q\in\overline\cQ_{\tilde{\pmb\xi}^k}^n}E_Q[y\xi^{k+1}+z\zeta]>\sup_{(a,b)\in(-\infty,\tilde\xi^{k+1}]\times[\pi^{k+1},\infty)}(ya+zb)\geq y\tilde\xi^{k+1}+z\pi^{k+1}(\zeta).$$
Obviously $y\geq 0$ and $z\leq 0$. There exists $\eps>0$ such that
$$y\tilde\xi^{k+1}+z(\pi^{k+1}(\zeta)-\eps)<\inf_{Q\in\overline\cQ_{\tilde{\pmb\xi}^k}^n}E_Q[y\xi^{k+1}+z\zeta].$$
Now consider the market with options $(\xi^1,\dotso,\xi^k,y\xi^{k+1}+z\zeta)$ available for static buying with prices $(\tilde\xi^1,\dotso,\tilde\xi^k,y\tilde\xi^{k+1}+z(\pi^{k+1}(\zeta)-\eps))$. It can be shown that SNA holds. Therefore, 
$$y\tilde\xi^{k+1}+z(\pi^{k+1}(\zeta)-\eps)\geq\inf_{Q\in\overline\cQ_{\tilde{\pmb\xi}^k}^n}E_Q[y\xi^{k+1}+z\zeta].$$
Contradiction.
\end{proof}

\begin{lemma}\label{l1}
SNA holds with $((\overline f^N,-\overline h^N),(\alpha,-\gamma))$ in the space $(\overline\Omega^N,\overline{\mathcal{F}}_T^N,\overline{\mathbb{F}}^N)$ if and only if SNA holds with $((\overline f^{N+1},-\overline h^{N+1}),(\alpha,-\gamma))$ in the space $(\overline\Omega^{N+1},\overline{\mathcal{F}}_T^{N+1},\overline{\mathbb{F}}^{N+1})$.
\end{lemma}
\begin{proof}
The proof is similar to the argument in \reref{r2} and we omit it here.
\end{proof}

\begin{lemma}\label{l11}
For $k=1,\dotso,K$, let $\mathfrak{g}^k=(\mathfrak{g}^k_t)_{t=0,\dotso,T}$ be a bounded $\overline{\mathbb{F}}^n$-adapted process uniformly continuous for $t=1,\dotso,T$, i.e.
 for $t=1,\dotso,T$ and $k=1,\dotso,K$,
$$\left|\mathfrak{g}^k_t(\omega^1,{\bf t} )-\mathfrak{g}^k_t(\omega^2,{\bf t})\right|\leq\rho\left(\max_{s=1,\dotso,t}|\omega_s^1-\omega_s^2|\right),\quad\omega^i=(\omega_1^i,\dotso,\omega_T^i)\in\Xi^T,\ i=1,2,\quad {\bf t}\in\T^n,$$
where $\rho$ is a modulus of continuity.
 Let $\mathcal{R}$ be a convex and weakly compact set of probability measures on $(\overline\Omega^n,\overline{\mathcal{F}}_T^n)$. Then
\begin{equation}\label{e73}
\sup_{\substack{\mu^k\in\overline\cL^n\\k=1,\dotso,K}}\inf_{R\in\cR}{E}_R\left[\sum_{k=1}^K\mu^k(\mathfrak{g}^k)\right]=\inf_{R\in\cR}\sup_{\substack{\mu^k\in\overline\cL^n\\k=1,\dotso,K}}{E}_R\left[\sum_{k=1}^K\mu^k(\mathfrak{g}^k)\right]=\inf_{R\in\cR}\sup_{\substack{\tau^k\in\overline\cT^n\\k=1,\dotso,K}}{E}_R\left[\sum_{k=1}^K \mathfrak{g}^k_{\tau^k}\right].
\end{equation}
\end{lemma}

\begin{proof}
We adapt the proof of \cite[Lemma 6.1]{ZZ8}. To overcome the difficulties stemming from both the i) discontinuity of stopping times, ii)  the fact that the collection $\cR$ may not have a dominating measure, we will first discretize $\overline\Omega^n$. We will then apply the minimax theorem to the discretized expression version of \eqref{e73}, and finally take the limit to conclude.

To this end, let $(A_i^m)_{i\in\mathbb{N}}\subset\mathcal{B}(\Xi)$ be a countable partition of $\Omega$, such that the diameter of each $A_i^m$ is less than $1/m$. Take $o_i^m\in A_i^m$ for each $i$ and $m$. Define the map $\lambda^m:\overline\Omega^n\mapsto\overline\Omega^n$ such that for any $\overline\omega=(\omega_1,\dotso,\omega_T,t^1,\dotso,t^n)\in\overline\Omega^n=\Xi^T\times\T^n$, if $\omega_t\in A_{i_t}^m$ for $t=1,\dotso,T$, then
$$\lambda^m(\overline\omega)=(o_{i_1},\dotso,o_{i_T},t^1,\dotso,t^n).$$
Let
$$\cR_m:=\{R\circ(\lambda^m)^{-1}:\ R\in\cR\}.$$
We shall proceed in four steps to show \eqref{e73}.

\textbf{Step 1}. We show that
\begin{equation}\label{e76}
\limsup_{m\rightarrow\infty}\sup_{\substack{\mu^k\in\overline\cL^n\\k=1,\dotso,K}}\inf_{R\in\cR_m}{E}_R\left[\sum_{k=1}^K\mu^k(\mathfrak{g}^k)\right]\leq\sup_{\substack{\mu^k\in\overline\cL^n\\k=1,\dotso,K}}\inf_{R\in\cR}{E}_R\left[\sum_{k=1}^K\mu^k(\mathfrak{g}^k)\right].
\end{equation}
Fix $\eps>0$. Let $(\mu_m^1,\dotso,\mu_m^K)\in\left(\overline\cL^n\right)^K$ be such that
$$\inf_{R\in\cR_m}{E}_R\left[\sum_{k=1}^K\mu_m^k(\mathfrak{g}^k)\right]\geq\sup_{\substack{\mu^k\in\overline\cL^n\\k=1,\dotso,K}}\inf_{R\in\cR_m}{E}_R\left[\sum_{k=1}^K\mu^k(\mathfrak{g}^k)\right]-\eps.$$
Define $(\tilde\mu_m^1,\dotso,\tilde\mu_m^K)$ by $(\tilde\mu_m^k)_t=(\mu_m^k)_t\circ\lambda^m$, for $t=0,\dotso,T$ and $k=1,\dotso,K$. We can show that $(\tilde\mu_n^1,\dotso,\tilde\mu_n^K)\in(\overline\cL^n)^K$. For any $\tilde R\in\cR$, let $\tilde R_m:=\tilde R\circ (\lambda^m)^{-1}\in\cR_m$. Then
$$E_{\tilde R_m}\left[\sum_{k=1}^K\mu_m^k(\mathfrak{g}^k)\right]=E_{\tilde R}\left[\sum_{k=1}^K\sum_{t=0}^T((\mu_m^k)_t\circ\lambda^m)(\mathfrak{g}^k_t\circ\lambda^m)\right]=E_{\tilde R}\left[\sum_{k=1}^K\sum_{t=0}^T(\tilde\mu_m^k)_t(\mathfrak{g}^k_t\circ\lambda^m)\right].$$
Therefore,
$$\left|{E}_{\tilde R_m}\left[\sum_{k=1}^K\mu_m^k(\mathfrak{g}^k)\right]-{E}_{\tilde R}\left[\sum_{k=1}^K\tilde\mu_m^k(\mathfrak{g}^k)\right]\right|\leq{E}_{\tilde R}\left[\sum_{k=1}^K\sum_{t=0}^T(\tilde\mu_m^k)_t\left|(\mathfrak{g}^k_t\circ\lambda^m)-\mathfrak{g}^k_t\right|\right]\leq K\rho(1/m).$$

Hence, we have that
\begin{equation*}
\begin{split}
\sup_{\substack{\mu^k\in\overline\cL^n\\k=1,\dotso,K}}\inf_{R\in\cR_m}{E}_R\left[\sum_{k=1}^K\mu^k(\mathfrak{g}^k)\right]-\eps&\leq\inf_{R\in\cR_m}{E}_R\left[\sum_{k=1}^K\mu_m^k(\mathfrak{g}^k)\right]\\
&\leq{E}_{\tilde R_m}\left[\sum_{k=1}^K\mu_m^k(\mathfrak{g}^k)\right]\leq{E}_{\tilde R}\left[\sum_{k=1}^K\tilde\mu_m^k(\mathfrak{g}^k)\right]+K\rho(1/m).
\end{split}
\end{equation*}
By the arbitrariness of $\tilde R$, we have that
\begin{equation*}
\begin{split}
\sup_{\substack{\mu^k\in\overline\cL^n\\k=1,\dotso,K}}\inf_{R\in\cR_m}{E}_R\left[\sum_{k=1}^K\mu^k(\mathfrak{g}^k)\right]-\eps&\leq\inf_{R\in\cR}{E}_R\left[\sum_{k=1}^K\tilde\mu_m^k(\mathfrak{g}^k)\right]+K\rho(1/m)\\
&\leq\sup_{\substack{\mu^k\in\overline\cL^n\\k=1,\dotso,K}}\inf_{R\in\cR}{E}_R\left[\sum_{k=1}^K\mu^k(\mathfrak{g}^k)\right]+K\rho(1/m).
\end{split}
\end{equation*}
Taking limsup on both sides above and then sending $\eps\searrow 0$, we have \eqref{e76} holds.

\textbf{Step 2}. We show that
$$\sup_{\substack{\mu^k\in\overline\cL^n\\k=1,\dotso,K}}\inf_{R\in\cR_m}{E}_R\left[\sum_{k=1}^K\mu^k(\mathfrak{g}^k)\right]=\inf_{R\in\cR_m}\sup_{\substack{\mu^k\in\overline\cL^n\\k=1,\dotso,K}}{E}_R\left[\sum_{k=1}^K\mu^k(\mathfrak{g}^k)\right].$$
As the domain of $\lambda^m$ is countable, there exists a probability measure $R^*$ on the domain of $\lambda^m$ that dominates $\cR_m$. Then we have that
\begin{eqnarray}
\notag&&\sup_{\substack{\mu^k\in\overline\cL^n\\k=1,\dotso,K}}\inf_{R\in\cR_m}{E}_R\left[\sum_{k=1}^K\mu^k(\mathfrak{g}^k)\right]=\sup_{\substack{\mu^k\in\overline\cL^n\\k=1,\dotso,K}}\inf_{R\in\cR_m}{E}_{R^*}\left[\frac{dR}{dR^*}\sum_{k=1}^K\mu^k(\mathfrak{g}^k)\right]\\
\notag&&=\inf_{R\in\cR_m}\sup_{\substack{\mu^k\in\overline\cL^n\\k=1,\dotso,K}}{E}_{R^*}\left[\frac{dR}{dR^*}\sum_{k=1}^K\mu^k(\mathfrak{g}^k)\right]=\inf_{R\in\cR_m}\sup_{\substack{\mu^k\in\overline\cL^n\\k=1,\dotso,K}}{E}_R\left[\sum_{k=1}^K\mu^k(\mathfrak{g}^k)\right],
\end{eqnarray}
where we apply the minimax theorem (see e.g., \cite[Corollary 2]{Frode}) for the second equality, and use the fact that $\overline\cL^n$ is compact and the map:
$$(\mu^1,\dotso,\mu^N)\mapsto{E}_{R^*}\left[\frac{dR}{dR^*}\sum_{k=1}^N\mu^k(\mathfrak{g}^k)\right]$$
is continuous under the Baxter-Chacon topology w.r.t. $R^*$ (see e.g., \cite{Edgar}).

\textbf{Step 3}. We show that
\begin{equation}\label{e77}
\inf_{R\in\cR}\sup_{\substack{\tau^k\in\overline\cT^n\\k=1,\dotso,K}}{E}_R\left[\sum_{k=1}^K \mathfrak{g}^k_{\tau^k}\right]\leq\liminf_{m\rightarrow\infty}\inf_{R\in\cR_m}\sup_{\substack{\tau^k\in\overline\cT^n\\k=1,\dotso,K}}{E}_R\left[\sum_{k=1}^K \mathfrak{g}^k_{\tau^k}\right].
\end{equation}
Without loss of generality, we assume the sequence $\left\{\inf_{R\in\cR_m}\sup_{\substack{\tau^k\in\overline\cT^n\\k=1,\dotso,K}}{E}_R\left[\sum_{k=1}^K \mathfrak{g}^k_{\tau^k}\right]\right\}_{m \in \mathbb{N}}$ converges. Fix $\eps>0$. Take $R_m\in\cR_m$ such that
$$\sup_{\substack{\tau^k\in\overline\cT^n\\k=1,\dotso,K}}{E}_{R_m}\left[\sum_{k=1}^K \mathfrak{g}^k_{\tau^k}\right]\leq\inf_{R\in\cR_m}\sup_{\substack{\tau^k\in\overline\cT^n\\k=1,\dotso,K}}{E}_R\left[\sum_{k=1}^K \mathfrak{g}^k_{\tau^k}\right]+\eps.$$
Let $\tilde R_m\in\cR$ be such that $R_m=\tilde R_m\circ(\lambda^m)^{-1}$. As $\cR$ is weakly compact, there exists $\tilde R\in\cR$ such that up to a subsequence $\tilde R_m\xrightarrow{w}\tilde R$. Then for any bounded uniformly continuous function $\f\in\mathcal{B}(\overline\Omega^n)$,
\begin{eqnarray}
\notag&&\hspace{-1cm}\left|{E}_{R_m}\f-{E}_{\tilde R}\f\right|\leq\left|{E}_{R_m}\f-{E}_{\tilde R_m}\f\right|+\left|{E}_{\tilde R_m}\f-{E}_{\tilde R}\f\right|=\left|{E}_{\tilde R_m}(\f\circ\lambda^m)-{E}_{\tilde R_m}\f\right|+\left|{E}_{\tilde R_m}\f-{E}_{\tilde R}\f\right|\\
\notag&&\leq{E}_{\tilde R_m}\left|(\f\circ\lambda^m)-\f\right|+\left|{E}_{\tilde R_m}\f-{E}_{\tilde R}\f\right|\leq\rho_\f(1/m)+\left|{E}_{\tilde R_m}\f-{E}_{\tilde R}\f\right|\rightarrow0,\quad m\rightarrow\infty,
\end{eqnarray}
where $\rho_\f$ is the modulus of continuity of $\f$. Hence, $R_m\xrightarrow{w}\tilde R$. Since the map
$$R\mapsto\sup_{\tau^k\in\overline\cT^n}{E}_R\left[\mathfrak{g}^k_{\tau^k}\right]$$
is lower semi-continuous under weak topology (see e.g., \cite[Theorem 1.1]{Elton89}), the map
$$R\mapsto\sum_{k=1}^K\sup_{\tau^k\in\overline\cT^n}{E}_R\left[\mathfrak{g}^k_{\tau^k}\right]=\sup_{\substack{\tau^k\in\overline\cT^n\\k=1,\dotso,K}}{E}_R\left[\sum_{k=1}^K \mathfrak{g}^k_{\tau^k}\right]$$
is also lower semi-continuous. Therefore,
\begin{eqnarray}
\notag&&\lim_{m\rightarrow\infty}\inf_{R\in\cR_m}\sup_{\substack{\tau^k\in\overline\cT^n\\k=1,\dotso,K}}{E}_R\left[\sum_{k=1}^K \mathfrak{g}^k_{\tau^k}\right]+\eps\geq\liminf_{m\rightarrow\infty}\sup_{\substack{\tau^k\in\overline\cT^n\\k=1,\dotso,K}}{E}_{R_m}\left[\sum_{k=1}^K \mathfrak{g}^k_{\tau^k}\right]\\
\notag&&\geq\sup_{\substack{\tau^k\in\overline\cT^n\\k=1,\dotso,K}}{E}_{\tilde R}\left[\sum_{k=1}^K \mathfrak{g}^k_{\tau^k}\right]\geq\inf_{R\in\cR}\sup_{\substack{\tau^k\in\overline\cT^n\\k=1,\dotso,K}}{E}_R\left[\sum_{k=1}^K \mathfrak{g}^k_{\tau^k}\right].
\end{eqnarray}
Letting $\eps\searrow0$ we obtain \eqref{e77}.

\textbf{Step 4}. From Steps 1-3 we have that
\begin{eqnarray}
\notag&&\sup_{\substack{\mu^k\in\overline\cL^n\\k=1,\dotso,K}}\inf_{R\in\cR}{E}_R\left[\sum_{k=1}^K\mu^k(\mathfrak{g}^k)\right]\leq\inf_{R\in\cR}\sup_{\substack{\mu^k\in\overline\cL^n\\k=1,\dotso,K}}{E}_R\left[\sum_{k=1}^K\mu^k(\mathfrak{g}^k)\right]\\
\notag&&=\inf_{R\in\cR}\sup_{\substack{\tau^k\in\overline\cT^n\\k=1,\dotso,K}}{E}_R\left[\sum_{k=1}^K \mathfrak{g}^k_{\tau^k}\right]\leq\liminf_{n\rightarrow\infty}\inf_{R\in\cR_m}\sup_{\substack{\tau^k\in\overline\cT^n\\k=1,\dotso,K}}{E}_R\left[\sum_{k=1}^K \mathfrak{g}^k_{\tau^k}\right]\\
\notag&&=\liminf_{m\rightarrow\infty}\inf_{R\in\cR_m}\sup_{\substack{\mu^k\in\overline\cL^n\\k=1,\dotso,K}}{E}_R\left[\sum_{k=1}^K\mu^k(\mathfrak{g}^k)\right]=\liminf_{m\rightarrow\infty}\sup_{\substack{\mu^k\in\overline\cL^n\\k=1,\dotso,K}}\inf_{R\in\cR_m}{E}_R\left[\sum_{k=1}^K\mu^k(\mathfrak{g}^k)\right]\\
\notag&&\leq\limsup_{m\rightarrow\infty}\sup_{\substack{\mu^k\in\overline\cL^n\\k=1,\dotso,K}}\inf_{R\in\cR_m}{E}_R\left[\sum_{k=1}^K\mu^k(\mathfrak{g}^k)\right]\leq\sup_{\substack{\mu^k\in\overline\cL^n\\k=1,\dotso,K}}\inf_{R\in\cR}{E}_R\left[\sum_{k=1}^K\mu^k(\mathfrak{g}^k)\right],
\end{eqnarray}
where the first and the second equalities (passage from the liquidation strategies to stopping times and then back) follow from \cite[Proposition 1.5]{Edgar}. 
%
\end{proof}

\begin{lemma}\label{l12}
Let \asref{a1}(i)(ii) hold. Then for $n=N,N+1$, the set
$$\overline\cQ^n(\alpha,\gamma):=\left\{Q\lll\overline\cP^n:\ \overline S^n\text{ is a $Q$-martingale},\ E_Q\left[\overline f^n\right]\leq\alpha,\ E_Q\left[\overline h^n\right]\geq\gamma\right\},$$
is weakly compact.
\end{lemma}
\begin{proof}
Take $\eps>0$. Since $\cQ(\alpha)$ is weakly compact, there exists a compact set $K\in\mathcal{B}(\Omega)$ such that $Q(K)\geq 1-\eps$ for any $Q\in\cQ(\alpha)$. For any $\overline Q\in\overline\cQ^n(\alpha,\gamma)$, we can write it as
\begin{equation}\label{e7}
\overline Q=Q'(\cdot)\otimes Q''(\cdot,t^1,\dotso,t^n),
\end{equation}
where $Q'$ is the marginal distribution of $\overline Q$ on $(\Omega,\mathcal{B}(\Omega))$, and $Q''$ is a transition kernel. Moreover, it is easy to see that $Q'\in\cQ(\alpha)$. Then we have that
$$\overline Q(K\times\mathbb{T}^n)=Q'(K)>1-\eps.$$
This implies that the set $\cQ^n(\alpha,\gamma)$ is tight and thus pre-compact.

Now take $(\overline Q_i)_{i=1}^\infty\subset\overline\cQ^n(\alpha,\gamma)$ such that $\overline Q_i\xrightarrow{w}\overline Q_\infty$, and we will show that $\overline Q_\infty\in\cQ^n(\alpha,\gamma)$. Denote as in \eqref{e7},
\begin{equation}\notag
\overline Q_i=Q_i'(\cdot)\otimes Q_i''(\cdot,t^1,\dotso,t^n),\quad i=1,\dotso,\infty.
\end{equation}
Obviously $Q_i'\xrightarrow{w}Q_\infty'$. As $\cQ(\alpha)$ is weakly compact, $Q_\infty'\in\cQ(\alpha)$. Then there exists $P\in\cP$ dominating $Q_\infty'$. This implies that $\overline Q_\infty\lll\overline\cP^n$, since $\overline Q_\infty$ is dominated by $P\otimes R$ for any probability measure $R$ on $\mathbb{T}^n$ with full support. Furthermore, $S$ is a $Q_\infty'$-martingale measure implies that $\overline S^n$ is a $\overline Q_\infty$-martingale,
and
$$E_{\overline Q_\infty}\left[\overline f^n\right]=E_{Q_\infty'}[f^n]\leq\alpha.$$
Finally, as $E_{\overline Q_i}\left[\overline h^n\right]\geq\gamma$ and $\overline Q_i\xrightarrow{w}\overline Q_\infty$, we have $E_{\overline Q_\infty}\left[\overline h^n\right]\geq\gamma$ by \asref{a1}(ii).
\end{proof}

\begin{proof}[\textbf{Proof of \thref{t3}}]
We adapt the proof of \cite[Theorem 5.1]{ZZ8}. As SNA holds in $\overline{\mathbb{M}}^N$, SNA$(\overline\cP^N)$ holds with $((\overline f^N,-\overline h^N),(\alpha,-\gamma))$. Then
\begin{eqnarray}
\notag \underline\pi(\phi)&=&\sup_{b\in\R_+^M}\sup_{\overline\mu\in\left(\overline\cL^N\right)^M,\overline\eta\in\overline\cL^N}\sup\Big\{x\in\R:\ \exists(\overline H,a,c)\in\overline\cH^N\times\R_+^L\times\R_+^N,\\
\notag&&\hspace{4cm}\text{ s.t. }\PhiN+\overline\eta(\overline\phi^N)\geq x,\ \overline\cP^N\text{-q.s.}\Big\}\\
\notag&=&\sup_{b\in\R_+^M}\sup_{\overline\mu\in\left(\overline\cL^N\right)^M,\overline\eta\in\overline\cL^N}\inf_{Q\in\overline\cQ^N(\alpha,\gamma)} E_Q\left[b\left(\overline\mu\left(\overline g^N\right)-\beta\right)+\overline\eta\left(\overline\phi^N\right)\right]\\
\label{e92} &=&\sup_{b\in\R_+^M}\inf_{Q\in\overline\cQ^N(\alpha,\gamma)}\sup_{\substack{\tau, \tau^j\in\overline\cT^N\\j=1,\dotso,M}} E_Q\left[\sum_{j=1}^M b^j\left(\overline {g_{\tau^j}^j}^N-\beta^j\right)+\overline\phi_\tau^N\right],
\end{eqnarray}
where we apply \leref{l2}(ii) for the second equality, and Lemmas \ref{l11} and \ref{l12} for the third equality. Now for $b\geq 0$ the map
$$b\mapsto\sup_{\substack{\tau, \tau^j\in\overline\cT^N\\j=1,\dotso,M}} E_Q\left[\sum_{j=1}^M b^j\left(\overline {g_{\tau^j}^j}^N-\beta^j\right)+\overline\phi_\tau^N\right]$$
is linear, and 
the map
$$Q\mapsto\sup_{\substack{\tau, \tau^j\in\overline\cT^N\\j=1,\dotso,M}} E_Q\left[\sum_{j=1}^M b^j\left(\overline {g_{\tau^j}^j}^N-\beta^j\right)+\overline\phi_\tau^N\right]$$
is convex and lower semi-continuous under the weak topology, see \cite[Theorem 1.1]{Elton89}. Thanks to the weak compactness of $\overline\cQ^N(\alpha,\gamma)$ by \leref{l12}, we can apply the minimax theorem (see e.g., \cite[Corollary 2]{Frode}) to \eqref{e92} and obtain
$$\underline\pi(\phi)=\inf_{Q\in\overline\cQ^N(\alpha,\gamma)}\sup_{b\in\R_+^M}\sup_{\substack{\tau, \tau^j\in\overline\cT^N\\j=1,\dotso,M}} E_Q\left[\sum_{j=1}^M b^j\left(\overline {g_{\tau^j}^j}^N-\beta^j\right)+\overline\phi_\tau^N\right]=\inf_{Q\in\overline\cQ^N(\alpha,\beta,\gamma)}\sup_{\tau\in\overline\cT^N}E_Q\left[\overline\phi_\tau^N\right].$$

Let $(Q_m)_{m\in\mathbb{N}}\subset\overline\cQ^N(\alpha,\beta,\gamma)\subset\overline\cQ^N(\alpha,\gamma)$. As $\overline\cQ^N(\alpha,\gamma)$ is weakly compact, there exist $Q\in\overline\cQ^N(\alpha,\gamma)$ and $(Q_{n_i})_{i\in\mathbb{N}}\subset(Q_n)_{n\in\mathbb{N}}$, such that $Q_{n_i}\xrightarrow{w} Q$. Again, by \cite[Theorem 1.1]{Elton89} the map $R\mapsto\sup_{\tau\in\overline\cT^N}\E_R\left[\overline {g_\tau^j}^N\right]$ is lower semi-continuous for $j=1,\dotso,M$. Therefore,
$$\sup_{\tau\in\overline\cT^N} E_Q\left[\overline {g_\tau^j}^N\right]\leq\liminf_{i\rightarrow\infty}\sup_{\tau\in\overline\cT^N} E_{Q_{n_i}}\left[\overline {g_\tau^j}^N\right]\leq\beta^j,\quad j=1,\dotso,M.$$
Hence, $Q\in\overline\cQ^N(\alpha,\beta,\gamma)$, which implies that $\cQ^N(\alpha,\beta,\gamma)$ is weakly compact. Therefore, the infimum for the duality of $\underline\pi(\phi)$ is attained since the map $R\mapsto\sup_{\tau\in\overline\cT^N}\E_R\left[\overline\phi_\tau^N\right]$ is lower semi-continuous.

As SNA$(\overline\cP^N)$ holds with $((\overline f^N,-\overline h^N),(\alpha,-\gamma))$, SNA$(\overline\cP^N)$ also holds with $((\overline f^{N+1},-\overline h^{N+1}),(\alpha,-\gamma))$ by \leref{l1}. Then corresponding results for the super-hedging price follows by similar but simpler arguments.
\end{proof}

\begin{proof}[\textbf{Proof of \thref{t4}}]
It is easy to show the sufficiency and let us focus on the necessity. We will prove by an induction on the number of longed American options, namely $M$. For $M=0$ the result follows from \leref{l2}(i). Now suppose the result holds for $M=m-1\in\mathbb{N}$ and show that it holds for $M=m$. For $k=m-1,m$, denote NA$^k$, SNA$^k$, $\underline\pi^k(\cdot)$ and $\overline\cQ^{N,k}$ as the NA defined in \deref{d2}, SNA defined in \deref{d2}, sub-hedging price defined in \eqref{e5}, and the set of martingale measures defined in \eqref{e29} in terms of $S, f, h$ and $g^1,\dotso,g^k$ in $\overline{\mathbb{M}}^N$, respectively. Denote $\pmb\beta^k:=(\beta^1,\dotso,\beta^k)$.

Let SNA$^m$ hold in $\overline{\mathbb{M}}^N$. Then there exists $\delta>0$, such that  NA$^m$ holds in $\overline{\mathbb{M}}^N$ w.r.t. $\alpha$, $(\beta^1,\dotso,\beta^{m-1},\beta^m-\delta)$, $\gamma$. It follows that
\begin{equation}\label{e91}
\underline\pi^{m-1}(g^m)\leq\beta^m-\delta,
\end{equation}
for otherwise, one would create an arbitrage by paying $\beta^m-\delta$ to buy one unit of $g^m$ and getting $(\underline\pi^{m-1}(g^m)+\beta^m-\delta)/2$ via some trading strategy.  As SNA$^m$ holds, SNA$^{m-1}$ also holds. Hence, by \thref{t3} we have that
\begin{equation}\label{e90}
\underline\pi^{m-1}(g^m)=\inf_{Q\in\overline\cQ^{N,m-1}(\alpha,\,\pmb\beta^{m-1},\gamma)}\sup_{\tau\in\overline{\mathcal{T}}^N} E_Q\left[\overline{g_\tau^m}^N\right].
\end{equation}
Moreover, by the induction hypotheses there exists $\delta'>0$ such that for any $P\in\overline\cP^N$, there exists $Q^1\in\overline\cQ^{N,m-1}(\alpha-\delta',\pmb\beta^{m-1}-\delta',\gamma+\delta')$ dominating $P$. 

By \asref{a1}(iii), there exists $C>0$ such that $|g_t^m|<C$ for $t=0,\dotso,T$. Choose $\lambda\in(0,1)$ such that
$$\beta^*:=\lambda C+(1-\lambda)(\beta^m-\delta/2)<\beta^m.$$
Let
$$\alpha':=\lambda(\alpha-\delta')+(1-\lambda)\alpha=\alpha-\lambda\delta',\quad\quad\gamma':=\gamma+\lambda\delta'$$
and
$$\beta':=(\beta^1-\lambda\delta',\dotso,\beta^{m-1}-\lambda\delta',\beta^*).$$

Let $P\in\cP$. We will show that there exists some $Q\in\overline\cQ^{N,m}(\alpha',\beta',\gamma')\subset\overline\cQ^{N,m}(\alpha-\eps,\beta-\eps,\gamma+\eps)$ dominating $P$, where $\eps:=(\beta^m-\beta^*)\wedge(\lambda\delta')$. By \eqref{e91} and \eqref{e90}, there exists $Q^2\in\overline\cQ^{N,m-1}(\alpha,\pmb\beta^{m-1},\gamma)$, such that
$$\sup_{\tau\in\overline\cT^N} E_{Q^2}\left[\overline{g_\tau^m}^N\right]<\beta^m-\delta/2.$$
Let
$Q_\lambda:=\lambda Q^1+(1-\lambda)Q^2\gg P.$
Obviously, $Q_\lambda\lll\cP$, $Q_\lambda$ is a martingale measure, $E_{Q_\lambda}\left[\overline f^N\right]\leq \alpha'$, $ E_{Q_\lambda}\left[\overline h^N\right]\geq\gamma'$, and
$$\sup_{\tau\in\overline\cT^N} E_{Q_\lambda}\left[\overline{g_\tau^j}^N\right]\leq\beta^j-\lambda\delta',\quad j=1,\dotso,m-1.$$
Furthermore,
$$\sup_{\tau\in\overline\cT^N} E_{Q_\lambda}\left[\overline{g_\tau^m}^N\right]=\sup_{\tau\in\overline\cT^N}\left(\lambda E_{Q^1}\left[\overline{g_\tau^m}^N\right]+(1-\lambda) E_{Q^2}\left[\overline{g_\tau^m}^N\right]\right)\leq\lambda C+(1-\lambda)\sup_{\tau\in\overline\cT^N} E_{Q^2}\left[\overline{g_\tau^m}^N\right]\leq\beta^*.$$
This implies $Q_\lambda\in\overline\cQ^{N,m}(\alpha',\beta',\gamma')$.
\end{proof}

{\small
\bibliographystyle{siam}
\bibliography{ref}}
\end{document}